\documentclass[11pt, a4paper]{article}
\usepackage[margin=1in]{geometry}
\usepackage{amsfonts, amsmath, amssymb, bm}
\usepackage{tabu}
\usepackage{enumerate}
\usepackage{tgcursor}
\usepackage{xcolor}
\usepackage{graphicx}
\usepackage{placeins}
\usepackage{listings}
\usepackage{stackengine}
\usepackage{caption}
\usepackage{subcaption}
\usepackage{setspace}
\usepackage{footnote}
\usepackage{comment}
\usepackage{footmisc}
\usepackage{hyperref}
\usepackage{eucal}
\usepackage{mathrsfs}
\usepackage{natbib}
\usepackage{authblk}
\usepackage[mathlines]{lineno} % line numbers, https://tex.stackexchange.com/questions/461186/how-to-use-lineno-with-
\usepackage{amsthm}

\usepackage{enumitem}

\newcommand*\linenomathpatch[1]{%
  \cspreto{#1}{\linenomath}%
  \cspreto{#1*}{\linenomath}%
  \csappto{end#1}{\endlinenomath}%
  \csappto{end#1*}{\endlinenomath}%
}
\newcommand*\linenomathpatchAMS[1]{%
  \cspreto{#1}{\linenomathAMS}%
  \cspreto{#1*}{\linenomathAMS}%
  \csappto{end#1}{\endlinenomath}%
  \csappto{end#1*}{\endlinenomath}%
}

\expandafter\ifx\linenomath\linenomathWithnumbers
  \let\linenomathAMS\linenomathWithnumbers
  \patchcmd\linenomathAMS{\advance\postdisplaypenalty\linenopenalty}{}{}{}
\else
  \let\linenomathAMS\linenomathNonumbers
\fi

\linenomathpatch{equation}
\linenomathpatchAMS{gather}
\linenomathpatchAMS{multline}
\linenomathpatchAMS{align}
\linenomathpatchAMS{alignat}
\linenomathpatchAMS{flalign}
\newtheorem{proposition}{Proposition}
\newtheorem{lemma}{Lemma}

\renewcommand{\vec}[1]{\boldsymbol{#1}}

\newcommand{\trp}{\,\prime}

\newcommand{\va}{\mathbf{a}}
\newcommand{\vA}{\mathbf{A}}

\newcommand{\vB}{\mathbf{B}}

\newcommand{\vC}{\mathbf{C}}

\newcommand{\vD}{\mathbf{D}}

\newcommand{\vE}{\mathbf{E}}

\newcommand{\vI}{\mathbf{I}}

\newcommand{\vm}{\mathbf{m}}
\newcommand{\vM}{\mathbf{M}}

\newcommand{\vN}{\mathbf{N}}

\newcommand{\vP}{\mathbf{P}}

\newcommand{\vR}{\mathbf{R}}
\newcommand{\vs}{\mathbf{s}}
\newcommand{\vS}{\mathbf{S}}
\newcommand{\vt}{\mathbf{t}}
\newcommand{\vT}{\mathbf{T}}
\newcommand{\vu}{\mathbf{u}}
\newcommand{\vU}{\mathbf{U}}
\newcommand{\vv}{\mathbf{v}}
\newcommand{\vV}{\mathbf{V}}
\newcommand{\vw}{\mathbf{w}}
\newcommand{\vW}{\mathbf{W}}
\newcommand{\vx}{\mathbf{x}}
\newcommand{\vX}{\mathbf{X}}

\newcommand{\vY}{\mathbf{Y}}

\newcommand{\vz}{\mathbf{z}}
\newcommand{\vZ}{\mathbf{Z}}

\newcommand{\vbeta}{\boldsymbol{\beta}}

\newcommand{\veta}{\boldsymbol{\eta}}

\newcommand{\vLambda}{\boldsymbol{\Lambda}}

\newcommand{\vOmega}{\boldsymbol{\Omega}}

\newcommand{\vPhi}{\boldsymbol{\Phi}}

\newcommand{\vrho}{\boldsymbol{\rho}}

\newcommand{\vSigma}{\boldsymbol{\Sigma}}

\newcommand{\vtheta}{\boldsymbol{\theta}}

\newcommand{\vxi}{\boldsymbol{\xi}}
\newcommand{\vXi}{\boldsymbol{\Xi}}

\hypersetup{
    colorlinks=true,
    linkcolor=blue,
    urlcolor=blue,
    citecolor=blue
}

\title{\vspace{-2cm}\LARGE\centering\normalfont{A flexible class of priors for orthonormal matrices with basis function-specific structure}}
\date{\vspace{-10mm}}

\author[1,$\ast$]{Joshua S. North}
\author[1]{Mark D. Risser}
\author[2]{F. Jay Breidt}
\affil[1]{Climate and Ecosystem Sciences Division, Lawrence Berkeley National Laboratory}
\affil[2]{Department of Statistics and Data Science, NORC at the University of Chicago}

\affil[*]{Corresponding author: jsnorth@lbl.gov}

\begin{document}

\maketitle

%\linenumbers

\abstract{
% One or two sentences providing a basic introduction to the field,  comprehensible to a scientist in any discipline. 
% The big data era of science and technology motivates s
Statistical modeling of high-dimensional matrix-valued data motivates the use of a low-rank representation that simultaneously summarizes key characteristics of the data and enables dimension reduction. % for data compression and storage.
% Two to three sentences of more detailed background, comprehensible to scientists in related disciplines.
Low-rank representations commonly factor the original data into the product of orthonormal basis functions and weights, where each basis function represents an independent feature of the data.
However, the basis functions in these factorizations are typically computed using algorithmic methods that cannot quantify uncertainty or account for basis function correlation structure \textit{a priori}.
% explicit structure beyond what is implicitly specified via data correlation.
While there exist Bayesian methods that allow for a common correlation structure across basis functions, empirical examples motivate the need for basis function-specific dependence structure.
% Additionally, the dependence structure of each basis function may operate on different scales such that enhanced flexibility is needed to properly capture the structure.
% One sentence summarizing the main result (with the words “here we show” or their equivalent).
We propose a prior distribution for orthonormal matrices that can explicitly model basis function-specific structure.
% Two or three sentences explaining what the main result reveals in direct comparison to what was thought to be the case previously, or how the main result adds to previous knowledge.
The prior is used within a general probabilistic model for singular value decomposition to conduct posterior inference on the basis functions while accounting for measurement error and fixed effects.
% To contextualize the proposed prior and corresponding model, we discuss how the prior specification can be used for various scenarios and relate the model to its deterministic counterpart.
We discuss how the prior specification can be used for various scenarios and demonstrate favorable model properties through synthetic data examples. 
Finally, we apply our method to two-meter air temperature data from the Pacific Northwest, enhancing our understanding of the Earth system's internal variability.

    \begin{center}
        \textit{Key Words: Bayesian Singular Value Decomposition, Probabilistic Low-Rank Representation, Probabilistic Basis Functions, Stiefel Manifold, Spatio-Temporal Random Effect}
    \end{center}
    
}

%%%%%%%%%%%%%%%%%%%%%%%%%%%%%%%%%%%%%%%%%%%%%%%%%%%%%%%%%%%%%%%%%%%%%%
\section{Introduction}
%%%%%%%%%%%%%%%%%%%%%%%%%%%%%%%%%%%%%%%%%%%%%%%%%%%%%%%%%%%%%%%%%%%%%%

%===================================
% Paragraph 1: 
%===================================
\subsection{Orthonormal matrices in statistical modeling}

Within the field of statistics, orthonormal matrices are the cornerstone of many modeling approaches, including exploratory data analysis, factor analysis \citep{Harman1976, Mulaik2009}, principal component analysis \citep[PCA;][]{Hotelling1933, Jolliffe2002}, singular value decomposition \citep[SVD;][]{Stewart1993}, and proper orthogonal decomposition \citep[POD;][]{Berkooz1993}.
Each of these techniques uses orthonormal matrices to decompose matrix-valued data with the goal of summarizing its key characteristics as well as dimension reduction \citep{Kambhatla1997} and data compression \citep{Chen2022}. 
Across many areas of science, technology, and medicine, orthonormal matrix factorizations of data are highly useful because the measurements of interest in these fields often arise from lower-dimensional processes with physically interpretable structures.
Examples include factor analysis in physiological studies \citep{Fabrigar1999}, PCA in geography \citep{Roden2015} and ecology \citep{Jackson1993, Peres-Neto2003}, and SVD and PCA for medical imaging \citep{Smith2014}.

%===================================
% Paragraph 2: 
%===================================
For mean-zero data $\vY \in \mathbb{R}^{n \times m}$, SVD decomposes $\vY = \vU \vD \vV'$, where $\vU \in \mathbb{R}^{n \times l}$ is an orthonormal matrix, $\vD \in \mathbb{R}^{l \times l}$ is a diagonal matrix, $\vV \in \mathbb{R}^{m \times l}$ is an orthonormal matrix, and $l = \text{min}\{n, m\}$. 
Alternatively, PCA decomposes $\vY \vY' = \vA \vB \vA'$, where now $\vA \in \mathbb{R}^{n \times l}$ is an orthonormal matrix whose columns are the eigenvectors of $\vY \vY'$, $\vB \in \mathbb{R}^{l \times l}$ is a diagonal matrix whose elements are the eigenvalues of $\vY \vY'$, and $l = \text{min}\{n, m\}$.
Note that the equivalence between SVD and PCA comes from $ \vY \vY'=\left(\vU \vD \vV'\right)\left( \vV \vD' \vU'\right)=\vU \vD \vD' \vU' = \vA \vB \vA'$, where the diagonal elements of $\vD$ are the square root of the eigenvalues of $\vY \vY'$, the columns of $\vU$ are the eigenvectors of $\vY \vY'$, and the columns of $\vV$ are the eigenvectors of $\vY' \vY$.

%===================================
% Paragraph 3:
%===================================
In the climate sciences where data are spatially- and temporally-oriented, the columns of orthonormal matrices define empirical orthogonal functions \citep[EOFs;][]{Lorenz1956, North1982, Hannachi2007}, which are analogous to PCA.
EOFs are used to summarize modes of climate variability \citep[see, e.g.,][]{Thompson2000,Mantua2002}, identify the drivers of extreme weather events \citep{Grotjahn2016}, and quantify human-induced changes to the global climate system \citep{OBrien2023}. 
Additionally, spatial modeling of climate data often uses EOFs to incorporate spatial and temporal information via spatially-indexed basis functions and spatial random effects \citep{Stroud2001, Nychka2002, Cressie2006, Cressie2008}.

%===================================
% Paragraph 4: Limitations of current modeling approaches
%===================================
\subsection{Inference and challenges}

The basis functions contained in the orthonormal matrices $\vU$ and $\vV$ and the elements of $\vD$ are traditionally computed via iterative methods \citep{Golub1965, Demmel1990}, which we refer to as classical SVD (C-SVD or C-PCA) henceforth.
However, these classical procedures have several important limitations.
First, when $n$ is large with respect to $m$, the basis functions contained in the orthonormal matrices estimated from C-SVD can be noisy and therefore lose their physical interpretation \citep{Wang2017}.
C-SVD and C-PCA are not able to distinguish between measurement and signal variation, which means that estimates of the basis functions are heavily influenced by the presence of measurement noise \citep{Bailey2012, Epps2019}.
Furthermore, since their algorithms are deterministic, C-PCA and C-SVD do not provide measures of uncertainty in either the basis functions or their weights.
Finally, the estimated basis functions only exhibit dependence or structure implicitly via data correlations since C-PCA and C-SVD cannot take advantage of explicit structure that may be present in the data generating mechanisms.

%===================================
% Paragraph 5: What has been done so far to address algorithmic limitations
%===================================
A variety of approaches have been developed to address limitations associated with C-SVD and C-PCA. 
Regarding the issue of noise, large $n$ with small $m$, and structure in the basis functions, a regularized PCA approach can be adopted \citep{Shen2008, Zou2006, Jolliffe2002b}.
\citet{Wang2017} extend the regularization approach by incorporating smoothness and local features into their penalization using smoothing splines and an $\ell_1$ penalty, producing spatially explicit orthogonal basis functions.
To further account for uncertainty quantification in the basis function, one possibility is to take a Bayesian approach and specify a prior distribution for the orthonormal matrix.
The set of orthonormal matrices $\mathcal{V}_{k,n}=\{\vX\in\mathbb{R}^{n\times k}:\vX' \vX = \vI_k\}$, where $\vI_k$ is the $k \times k$ identity matrix, is called the Stiefel manifold \citep{Chikuse2003}.
Considerable effort has been put into understanding theoretical properties associated with distributions on the Stiefel manifold and optimal methods for computation and sampling \citep{Mardia1999, Chikuse2003, Hoff2007, Hoff2009, Byrne2013, Wang2013, Wang2014, Hernandez-Stumpfhauser2017, Pourzanjani2021, Jauch2021}.
\citet{Hoff2007} developed a uniform prior for orthonormal basis functions (the invariant or uniform measure on the Stiefel manifold) that enables the specification of a Bayesian SVD model, and showed how to sample from the full conditional distributions of the model.
However, the approach in \citet{Hoff2007} requires sampling from the von Mises-Fisher (or Bingam-von Mises-Fisher) distributions, which can be difficult, and does not allow for the basis functions to be structured.
Additionally, support for these distributions in probabilistic programming languages such as Stan is limited \citep{Carpenter2017}, providing yet another barrier for implementation.
\cite{Hoff2009} and \cite{Byrne2013} propose tractable methods for sampling from von Mises-Fisher distributions, but these require the underlying statistical model to abide by specific conditions and forms which limits the application areas.
Recent work by \citet{Pourzanjani2021} and \citet{Jauch2021} addresses both sampling and flexibility of distributions on the Stiefel manifold by simulating unconstrained random vectors (i.e., not orthogonal and not unit-length) and then transforming these draws to be orthonormal via an appropriate Jacobian to obtain samples on the Stiefel manifold.
Importantly, these methods are computationally efficient, can be incorporated into probabilistic programming languages, and allow for the basis functions to be modeled dependently.
However, the dependence structure is limited in that it is shared across the basis functions and is unable to accommodate the basis function-specific structures that are present in real-world data sets. % (see, e.g., Figure~\ref{fig:lengthScaleMotivation}).

%===================================
% Paragraph 6: Example motivating what is missing
%===================================
Particularly in the climate sciences, the physical structures summarized by orthonormal matrices have different scales (e.g., spatial or temporal), wherein the leading modes or basis functions reflect larger-scale variability while the later modes reflect finer-scale variability.
To illustrate this, we calculated the SVD of monthly maximum two-meter air temperature from a $0.25^\circ \times 0.25^\circ$ longitude-latitude grid over the United States Pacific Northwest from 1979 through 2022 (see Section \ref{sec:SAT} for details on the data) using standard statistical software.
We then estimate the length-scale of a Gaussian variogram for each spatial and temporal basis function, the columns of the left- and right- singular matrices, respectively.
Figure \ref{fig:lengthScaleMotivation} shows empirical estimates of the length-scale for each basis function for $\vU$ and $\vV$ in panels a) and b), respectively.
From this figure it is clear the length-scale of the leading modes for both the left- and right- singular matrices is at least one order of magnitude larger than that of the later modes, following a quasi-exponentially decreasing trend. 
This suggests estimating a common spatial or temporal structure for all of the basis functions will miss important features of the data, resulting in oversmoothing and underfitting for the leading modes and undersmoothing and overfitting for the later modes.

\begin{figure}[!t]
    \centering
    \includegraphics[width = 0.95\linewidth]{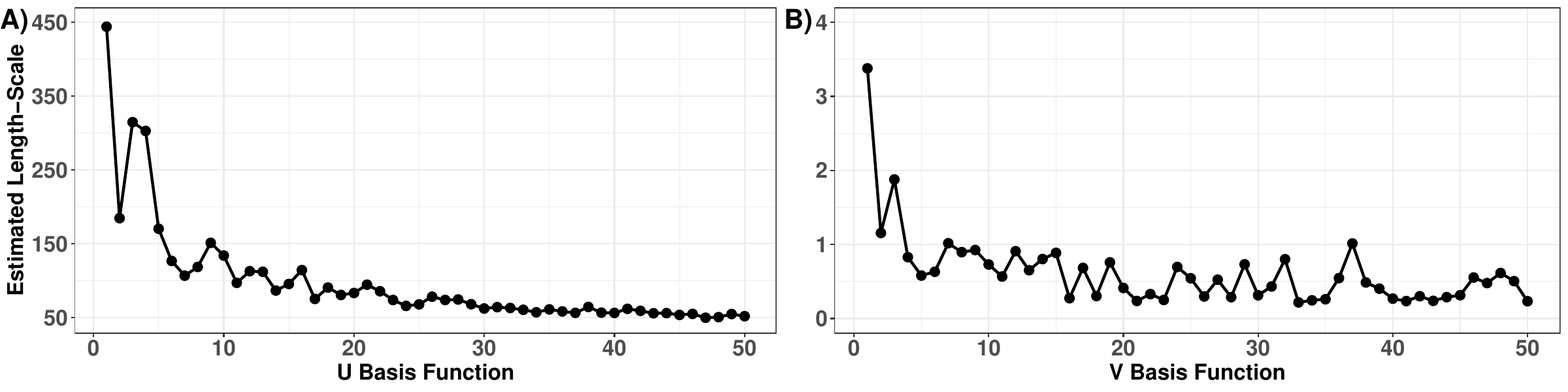}
    \caption{Estimated length-scale from a fitted Gaussian variogram for each spatial and temporal basis function computed from the singular value decomposition of the two-meter air surface temperature data described in Section \ref{sec:SAT}.}
    \label{fig:lengthScaleMotivation}
\end{figure}

%===================================
% Point-of-paragraph:
% What are we doing that's new
%===================================
\subsection{Contributions}

Here, we develop a prior distribution for orthonormal matrices that enables basis function specific structure and construct a probabilistic model for SVD.
The resulting full conditional distributions for the basis functions are available in closed form, yielding an analytically straightforward posterior for sampling orthonormal matrices.
Furthermore, we discuss how the prior can be used for a variety of modeling purposes.
Our prior is in general not uniformly distributed on $\mathcal{V}_{k,n}$ (although the uniform distribution is a special case) and we are able to impart information into the prior through our specification of a correlation matrix.
We show how the correlation matrix can be specified to either impart smoothing onto the basis functions \citep[producing results similar to][]{Wang2017} or recover the prior developed by \citet{Hoff2007}, and also demonstrate how the mean of the full conditional distributions for the basis functions of our probabilistic SVD model coincides with the classical approach (C-SVD) under certain conditions.
Our resulting prior, along with the proposed Bayesian hierarchical model, is quite general and allows each basis function to have a unique dependence structure that is learned from the data, which has not been previously possible.

%===================================
% Point-of-paragraph:
% Paper organization
%===================================
The remainder of the manuscript is organized as follows. 
Section~\ref{sec:Prior} develops the prior distribution for orthonormal matrices.
Section~\ref{sec:GPM} proposes a general probabilistic model for matrix factorizations with a specific focus on SVD and then expands other possible modeling choices.
Three simulation studies are conducted in section~\ref{sec:dataExamples}, where we show the importance of basis function-specific structure, the model rank and signal-to-noise ratios, and the impact a linear trend has on basis function recovery.
In Section~\ref{sec:SAT}, we apply our probabilistic model for SVD to decompose monthly maximum two-meter air temperature into its major modes of variability and provide uncertainty bounds for these modes, allowing better understanding of the spatial relationships in the data and illustrating the importance of basis function-specific structure.
Section~\ref{sec:discussion} concludes the paper.

%%%%%%%%%%%%%%%%%%%%%%%%%%%%%%%%%%%%%%%%%%%%%%%%%%%%%%%%%%%%%%%%%%%%%%
\section{A prior distribution for orthonormal matrices with basis function-specific structure}\label{sec:Prior}
%%%%%%%%%%%%%%%%%%%%%%%%%%%%%%%%%%%%%%%%%%%%%%%%%%%%%%%%%%%%%%%%%%%%%%

We construct a prior distribution for matrices on the Stiefel manifold $\mathcal{V}_{k,n}$ that is conjugate with a normal likelihood model.
The prior is constructed from the projected normal distribution that has been augmented with a latent length (see \citealt{Wang2013, Wang2014} and \citealt{Hernandez-Stumpfhauser2017} and the references therein for details on the projected normal).

%%%%%%%%%%%%%%%%%%%%%%%%%%%%%%%%%%%%%%%%%%%%%%%%%%%%%%%%%%%%%%%%%%%%%%
\subsection{Generating mechanism}\label{sec:GeneratingMechanism}

One method of drawing an orthonormal matrix from the uniform distribution on $\mathcal{V}_{k,n}$ is outlined in the appendix of \citet{Hoff2007}.
As part of the construction, the underlying normal distribution from which the orthonormal matrix is generated specifies the identity matrix as the covariance, and the resulting distribution is uniform on $\mathcal{V}_{k,n}$.
Here, we extend this generating mechanism to allow for structure in its covariance, specific to each column, such that the prior implied by \cite{Hoff2007} is a special case.
By construction, the resulting distribution is not necessarily uniform on $\mathcal{V}_{k,n}$.

For fixed $k$, let $\vz_i$ independent $ \mbox{N}_n(\vec{0}, \vOmega_i)$ and $\vOmega_i \sim \pi_{\Omega}$, for $i=1,2,\ldots,k$, where $\pi_{\Omega}$ is a valid distribution for symmetric positive definite matrices.
Define $\vP_0=\vI_n$, $\vx_1=\vP_0\vz_1$, and 
\begin{align*}
    \vX_i= [\vx_1, \ldots, \vx_i], \quad \vP_i= \vI_n - \vX_i (\vX'_i \vX_i)^{-1} \vX'_i, \quad \vx_{i+1}=\vP_i\vz_{i+1}
\end{align*}
for $i=1,2,\ldots,k-1$.
Then $\vx_i|\vX_{i-1} \sim \mbox{N}_n(\vec{0}, \vP_{i-1}\vOmega_i\vP'_{i-1})$ and $\vx'_i \vx_j = 0$ for $i \neq j$.
Further, define
\begin{align}\label{eqn:orthomatrix}
    \vw_i=\frac{\vx_i}{(\vx_i'\vx_i)^{1/2}},\quad \vW_i=[\vw_1,\ldots,\vw_i]
\end{align}
for $i=1,2,\ldots,k$.
By construction, $\vW_k \in \mathcal{V}_{k,n}$ is an orthonormal matrix.
The conditional distributions of each column given the preceding columns are $\vw_i |\vW_{i-1} \sim \mbox{PN}_n(\vec{0}, \vP_{i-1}\vOmega_i\vP'_{i-1})$, where $\mbox{PN}_n(\cdot, \cdot)$ denotes the $n$-dimensional projected normal distribution \citep{Wang2013, Wang2014, Hernandez-Stumpfhauser2017}.

Let $\overset{d}{=}$ denote equality in distribution.
We now provide two key properties associated with the distribution of $\vW \equiv \vW_k$ based on the constructed matrix $\vX \equiv \vX_k$, with proofs deferred to appendix~\ref{sec:props}. 

%%%%%%%%%%%%%%%%%%%%%%%%%%%%%%%%%%%%%%%%%%%%%%%%%%%%%%%%%%%%%%%%%%%%%%
\begin{proposition}\label{prop:one}
    The columns of $\vW = \vX (\vX' \vX)^{-1/2}$ are exchangeable. That is, for any permutation $\pi$ of the set $\{1, \ldots, k\}$, $p([\vw_1, \ldots, \vw_k]) \overset{d}{=} p([\vw_{\pi(1)}, \ldots, \vw_{\pi(k)}])$.
\end{proposition}
%%%%%%%%%%%%%%%%%%%%%%%%%%%%%%%%%%%%%%%%%%%%%%%%%%%%%%%%%%%%%%%%%%%%%%

%%%%%%%%%%%%%%%%%%%%%%%%%%%%%%%%%%%%%%%%%%%%%%%%%%%%%%%%%%%%%%%%%%%%%%
\begin{proposition}\label{prop:two}
    $\vw_i|\vW_{i-1} \overset{d}{=} \vN_{i-1}\widetilde{\vw}_i|\vW_{i-1}$ where the columns of $\vN_{i-1}$ form an orthonormal basis for the null space of $\vW_{i-1}$ and $\widetilde{\vw}_i$, the projected weight function, satisfies $\widetilde{\vw}_i|\vW_{i-1} \sim \emph{PN}_{n-i+1}(\vec{0}, \vN_{i-1}'\vOmega_i\vN_{i-1})$.
\end{proposition}
%%%%%%%%%%%%%%%%%%%%%%%%%%%%%%%%%%%%%%%%%%%%%%%%%%%%%%%%%%%%%%%%%%%%%%
Proposition~\ref{prop:one} implies the columns of $\vW$ are exchangeable, and therefore the conditional distribution $\vw_i|\vW_{Q}$ is invariant to the choice of subset of columns $Q \subset \{1, \ldots, k\}$.
When proposition~\ref{prop:one} is taken with proposition~\ref{prop:two}, the conditional distribution of $\vw_i|\vW_{Q}$ given any subset of columns $Q$ is equal in distribution to $\vN_{Q}\widetilde{\vw}_i$, where $\vN_{Q}$ is an orthonormal basis for the null space of $\vW_{Q}$ and $\widetilde{\vw}_i|\vW_{Q} \sim \mbox{PN}_{n-|Q|+1}(\vec{0}, \vN'_{Q}\vOmega_i\vN_{Q})$.
Therefore, we now focus on a prior distribution for $\widetilde{\vw}_i$, the \textit{projected weight function}, where $\widetilde{\vw}_i|\vW_{-i} \sim \mbox{PN}_{n-k+1}(\vec{0}, \vN'_{i}\vOmega_i\vN_{i})$ (i.e., $Q = \{1, \ldots, i-1, i+1, \ldots, k\}$) and the columns of $\vN_{i}$ span the null space of $\vW_{-i}$.

%%%%%%%%%%%%%%%%%%%%%%%%%%%%%%%%%%%%%%%%%%%%%%%%%%%%%%%%%%%%%%%%%%%%%%
\subsection{Projected normal prior distribution}\label{sec:PriorDistribution}

From the construction in Section~\ref{sec:GeneratingMechanism}, we have $\widetilde{\vw}_i|\vW_{-i} \sim \mbox{PN}_{n-k+1}(\vec{0}, \vN_{i}'\vOmega_i\vN_{i})$.
However, sampling from a high-dimensional projected normal distribution is difficult because of the form of the density function.
To make sampling from the projected normal tractable, we augment the distribution $\widetilde{\vw}_i|\vW_{-i}$ using a latent length variable $r_i$.
The joint distribution of $(r_i, \widetilde{\vw}_i) | \vW_{-i}$ can be derived by transforming the random variable $\vx_i$ to spherical coordinates (see supplement~\ref{sec:PND}), where the density function is
\begin{align}\label{eqn:PN_density}
    p(r_i, \widetilde{\vw}_i|\vW_{-i}) = (2 \pi)^{-n^*/2}|\vN_{i}'\vOmega_i\vN_{i}|^{-1/2}\exp\left\{-\frac{1}{2}(r_i \widetilde{\vw}_i)'(\vN_{i}'\vOmega_i\vN_{i})^{-1}(r_i \widetilde{\vw}_i) \right\}r_i^{n^*-1}\mathbb{I}(\widetilde{\vw}_i \in \mathcal{V}_{1,n^*}),
\end{align}
which we denote as $p(r_i, \widetilde{\vw}_i) \sim \mbox{N}_{n^*}(\vec{0}, \vN_{i}'\vOmega_i\vN_{i})r_i^{n^*-1}$ with $n^* = n-k+1$.
The indicator function $\mathbb{I}(\widetilde{\vw}_i \in \mathcal{V}_{1,n^*})$ is an integrating constant that is independent of the angle of $\widetilde{\vw}_i$ and dependent only on its length.
Note for $k=1$, the Stiefel manifold $\mathcal{V}_{1,n}$ is the $n-1$-dimensional unit sphere and $\mathcal{V}_{1,n} \equiv \mathbb{S}^{n-1}$.
The length variable $r_i$ can be sampled using either a slice sampler \citep{Hernandez-Stumpfhauser2017} or a Metropolis-Hastings algorithm.
However, we have found the slice sampler has numerical issues when $n$ is large, and use a Metropolis-Hastings within Gibbs algorithm (see supplement~\ref{sec:FCD}) for all examples presented herein.

The PN prior is convenient because if the data distribution is normal, the resulting full conditional distribution is proportional to a normal, which is easy to sample from (see Section~\ref{sec:PSVD} and supplement~\ref{sec:FCD} for more detail).

%%%%%%%%%%%%%%%%%%%%%%%%%%%%%%%%%%%%%%%%%%%%%%%%%%%%%%%%%%%%%%%%%%%%%%
\subsection{Incorporating explicit structure into the prior}\label{sec:PriorStructure}

From our formulation of the prior, we have the ability to specify or estimate the correlation structure for the projected basis functions.
The non-informative choice is $\vOmega_i \propto \vI_n$, implying there is no dependence between the elements of the basis functions.
As discussed in the supplement (\ref{sec:IC}), when $\vOmega_i \equiv \vI$ the generating mechanism is equivalent to that proposed by \citet{Hoff2007}, resulting in $\widetilde{\vw}_i$ being distributed uniformly on the $(n-k+1)$-dimensional sphere and the prior being equivalent to \citet{Hoff2007}.

A more general choice is to model $\vOmega_i = \sigma^2_{i}\vC_i$, where $\vC_i$ is a positive-definite correlation matrix that specifies structure among the elements in the $i$th basis function and $\sigma^2_i$ is a common variance parameter for those elements.
(While $\sigma^2_i$ does not impact the distribution of $\widetilde{\vw}_i$ or $\vw_i$ directly because they are of unit length, it does affect the joint distribution (\ref{eqn:PN_density}) of $(r_i, \widetilde{\vw}_i)$.)
In most cases, $\vC_i \equiv \vC(\vtheta_i)$ will depend on hyperparameters $\vtheta_i$ that can either be specified or learned within the hierarchical model.
Across many areas of science, including spatial statistics, machine learning, and emulation of complex physical models, the elements of $\vC_i$ are modeled via kernel functions $C_\theta(\cdot, \cdot)$ that are positive definite on the domain specified by the input space $\mathcal{S}$.
For example, when $\mathcal{S} \subset \mathbb{R}^d$, a popular choice is the Mat\'ern kernel
\begin{align}\label{eqn:matern}
    C_{\nu, \rho}(\vs, \vs') = \frac{2^{1-\nu}}{\Gamma(\nu)}\left(2\nu \frac{||\vs-\vs'||}{\rho}\right)^{\nu}J_{\nu}\left(2\nu \frac{||\vs-\vs'||}{\rho}\right),
\end{align}
defined for $\vs,\vs' \in \mathcal{S}$,
where $\Gamma$ is the gamma function, $J_{\nu}$ is the Bessel function of the second kind, and $\vtheta = \{\nu, \rho\}$ are hyperparameters that describe the differentiability and length-scale of the implied stochastic process, respectively. 
Special cases of the Mat\'ern kernel are for $\nu = 0.5$, in which (\ref{eqn:matern}) simplifies to the exponential kernel $C_{0.5, \rho}(\vs, \vs') = \exp\{-||\vs - \vs'||/\rho\}$, and the limit as $\nu\rightarrow\infty$, in which (\ref{eqn:matern}) reduces to the squared exponential or Gaussian kernel $C_{\infty, \rho}(\vs, \vs') = \exp\{-||\vs - \vs'||^2/\rho\}$.
Kernel functions like the Mat\'ern are useful for modeling generic dependence because they are highly flexible, depend on only a few hyperparameters (each of which is interpretable), yield data-driven smoothing that can characterize nonlinear structures in the underlying data, and require minimal \textit{a priori} or subjective specification.
Furthermore, such kernel functions do not require offline tuning of bandwidth or regularization parameters \citep[as is needed in, e.g., smoothing splines; see][]{Wang2017} since these aspects of the kernel can be inferred from the data within the Bayesian hierarchical model.

%%%%%%%%%%%%%%%%%%%%%%%%%%%%%%%%%%%%%%%%%%%%%%%%%%%%%%%%%%%%%%%%%%%%%%
\section{General probabilistic model}\label{sec:GPM}
%%%%%%%%%%%%%%%%%%%%%%%%%%%%%%%%%%%%%%%%%%%%%%%%%%%%%%%%%%%%%%%%%%%%%%

Define $\vZ \in \mathbb{R}^{n \times m}$ to be the observed data which is modeled as
\begin{align}\label{eqn:data}
    \vZ = \vM + \vY + \vA\vXi\vB,
\end{align}
where $\vM \in \mathbb{R}^{n \times m}$, $\vY \in \mathbb{R}^{n \times m}$, $\vSigma = \vA \vA' \in \mathbb{R}^{n \times n}$, $\vPhi = \vB \vB' \in \mathbb{R}^{m \times m}$, and $\vXi = [\vxi_1, \ldots, \vxi_m]$ with $\vxi_i$ independent $N_n(0, \vI_m)$ for $i = 1, \ldots, m$.
Then $\vZ|\vM, \vY, \vSigma, \vPhi \sim \mbox{MN}_{n \times m}(\vM + \vY, \vSigma, \vPhi)$ where $\mbox{MN}$ is the matrix normal distribution, $\vM + \vY$ is the mean of $\vZ$, $\vSigma$ is the covariance matrix for the rows of $\vZ$, $\vPhi$ is the covariance matrix for the columns of $\vZ$, and the density function is
\begin{align}\label{eqn:SVD_M_term}
    p(\vZ|\vM, \vY, \vSigma, \vPhi) = \frac{1}{(2\pi)^{nm/2}|\vPhi|^{n/2}|\vSigma|^{m/2}}\exp\left\{-\frac{1}{2} \text{tr}\left[ \vPhi^{-1}(\vZ - \vM - \vY)'\vSigma^{-1}(\vZ - \vM - \vY) \right] \right\}.
\end{align}
Equation (\ref{eqn:data}) is a mixed-effects model, where $\vM$ is a fixed-effect mean structure that is dependent on observed covariates, which we discuss in Section~\ref{sec:LinearTrend}, and $\vY$ is a ``smooth'' random effect that we will represent using basis functions and weights.
Generally, we assume $\vY$ is a mean zero random effect and explains any discrepancy in $\vZ$ not explained by $\vM$.
For example, if $\vZ$ is oriented such that the rows index spatial locations and the columns index temporal observations (or replications), $\vY$ would be considered spatial random effects.
We now specify a non-parametric model for the random effects $\vY$ using singular value decomposition and the prior distribution proposed in Section~\ref{sec:PriorDistribution}.

%%%%%%%%%%%%%%%%%%%%%%%%%%%%%%%%%%%%%%%%%%%%%%%%%%%%%%%%%%%%%%%%%%%%%%
\subsection{A probabilistic model for singular value decomposition}\label{sec:PSVD}
%%%%%%%%%%%%%%%%%%%%%%%%%%%%%%%%%%%%%%%%%%%%%%%%%%%%%%%%%%%%%%%%%%%%%%

For now, we assume the mean of $\vZ$ is zero (i.e., $\vM = \vec{0}$) and focus on a model for $\vY$.
In models such as (\ref{eqn:data}), the process $\vY$ can be represented as a reduced-rank process.
One example of a reduced-rank model is the singular value decomposition (SVD) $\vY = \vU \vD \vV'$, where $\vU \in \mathbb{R}^{n \times l}$ is an orthonormal matrix, $\vD \in \mathbb{R}^{l \times l}$ is a diagonally structured matrix, $\vV \in \mathbb{R}^{m \times l}$ is an orthonormal matrix, and $l = \text{min}\{n, m\}$.
To reduce the dimension of the process, we set $k < l$ (typically $k \ll l$) where $k$ is some pre-specified value.
This results in $\vY \approx \vU \vD \vV'$, where now $\vU \in \mathbb{R}^{n \times k}$, $\vD \in \mathbb{R}^{k \times k}$, and $\vV \in \mathbb{R}^{m \times k}$ are of reduced dimension.
In traditional SVD (similarly in PCA), the amount of variance explained by each component can be used to inform the value of $k$.
For now, we will assume $k$ is fixed and refer the reader to Section~\ref{sec:ModelSpecifications} for further discussion.

In (\ref{eqn:data}), $\vPhi = \vB \vB'$ represents the covariance between replicate observations (columns).
We make the simplifying assumption $\vPhi = \vI_m$ (i.e., independence between replicates) and model all variation in the data through $\vSigma$, which represents the covariance within observations (rows).
The resulting probability model is $\vZ \sim \mbox{MN}_{n \times m}(\vU \vD \vV', \vSigma, \vI_m)$, where $\vSigma$ now accounts for the approximation of choosing $k \ll l$ and the density function is
\begin{align}\label{eqn:SVD_data}
    p(\vZ|\vU, \vD, \vV, \vSigma, \vI_m) = \frac{1}{(2\pi)^{nm/2}|\vSigma|^{m/2}}\exp\left\{-\frac{1}{2} \text{tr}\left[(\vZ - \vU \vD \vV')'\vSigma^{-1}(\vZ - \vU \vD \vV') \right] \right\}.
\end{align}

%%%%%%%%%%%%%%%%%%%%%%%%%%%%%%%%%%%%%%%%%%%%%%%%%%%%%%%%%%%%%%%%%%%%%%
\subsubsection{Model priors}

To complete our model specification, we assign priors to $\vU, \vD, \vV,$ and $\vSigma$, and estimate the model parameters using Bayesian techniques.
Define $\vU_{-i} \equiv [\vu_1, \ldots, \vu_{i-1}, \vu_{i+1}, \ldots, \vu_k]$, $\vV_{-i} \equiv [\vv_1, \ldots, \vv_{i-1}, \vv_{i+1}, \ldots, \vv_k]$, $\vD_{-i} \equiv \text{diag}(d_1, \ldots, d_{i-1}, d_{i+1}, \ldots, d_{k})$, and $\vE_{-i} \equiv \vZ - \vU_{-i}\vD_{-i}\vV'_{-i}$, so that $\vZ - \vU\vD\vV' = \vE_{-i} - d_i\vu_i\vv'_i$.
Factoring the trace of the exponent of (\ref{eqn:SVD_data}),
\begin{align*}
    \text{tr}[(\vZ - \vU\vD\vV')'\vSigma^{-1}(\vZ - \vU\vD\vV')] & = \text{tr}[(\vE_{-i} - d_i\vu_i\vv'_i)'\vSigma^{-1}(\vE_{-i} - d_i\vu_i\vv'_i)] \\
            & = \text{tr}[\vE'_{-i}\vSigma^{-1}\vE_{-i} - 2d_i\vv_i\vu'_i\vSigma^{-1}\vE_{-i} + d^2_i\vv_i\vu'_i\vSigma^{-1}\vu_i\vv'_i] \\
            & = \text{tr}[\vE'_{-i}\vSigma^{-1}\vE_{-i} - 2d_i\vu'_i\vSigma^{-1}\vE_{-i}\vv_i + d^2_i\vv'_i\vv_i\vu'_i\vSigma^{-1}\vu_i] \\
            & = \text{tr}[\vE'_{-i}\vSigma^{-1}\vE_{-i} - 2d_i\vu'_i\vSigma^{-1}\vE_{-i}\vv_i + d^2_i\vu'_i\vSigma^{-1}\vu_i].
\end{align*}
The distribution $\vZ \sim \mbox{MN}_{n \times m}(\vU \vD \vV', \vSigma, \vI_m)$ can then be written
\begin{align}\label{eqn:SVDfactored}
    p(\vZ|\vu_i, \vv_i, d_i, \vU_{-i}, \vD_{-i}, \vV_{-i}, \vSigma) = & \\ \frac{1}{(2\pi)^{nm/2}|\vSigma|^{m/2}}\exp\biggl\{-\frac{1}{2} \text{tr}\Bigl[&\vE'_{-i}\vSigma^{-1}\vE_{-i} - 2d_i\vu'_i\vSigma^{-1}\vE_{-i}\vv_i + d^2_i\vu'_i\vSigma^{-1}\vu_i\Bigr] \biggr\} \nonumber,
\end{align}
which enables inference on the columns of $\vU$ and $\vV$ and the elements of $\vD$ individually (e.g., inference on $\vu_i$ and $\vv_i$).
Recall from Section~\ref{sec:PriorDistribution} that $\vu_i|\vU_{-i} \overset{d}{=} \vN^u_{i} \widetilde{\vu}_i|\vU_{-i}$ and $\vv_i|\vV_{-i} \overset{d}{=} \vN^v_{i} \widetilde{\vv}_i|\vV_{-i}$ where the columns of $\vN^u_{i}$ and $\vN^v_{i}$ span the null space of $\vU_{-i}$ and $\vV_{-i}$, respectively.
We specify the prior distributions
\begin{align}\label{eqn:SVDpriors}
\begin{split}
    d_{i}\widetilde{\vu}_i|\vU_{-i} & \sim \mbox{N}_{n-k+1}(\vec{0}, \vN^{u \trp}_{i}\vOmega^u_i\vN^u_{i})d_{i}^{n-k}\mathbb{I}(\widetilde{\vu}_i \in \mathcal{V}_{1,n-k+1}) \\
    d_{i}\widetilde{\vv}_i|\vV_{-i} & \sim \mbox{N}_{m-k+1}(\vec{0}, \vN^{v \trp}_{i}\vOmega^v_i\vN^v_{i})d_{i}^{m-k} \mathbb{I}(\widetilde{\vv}_i \in \mathcal{V}_{1,m-k+1}) \\
    d_i & \sim\mbox{Unif}(0, \infty).
\end{split}
\end{align}
For simplicity, we assume $\vSigma = \sigma^2\vI_n$, but this simplification can be relaxed if desired, e.g., by allowing $\vSigma$ to be a structured non-diagonal covariance matrix.
Last, we specify $\vOmega^u_i = \sigma^2_{u,i}\vC_u(\vtheta_{u,i})$ and $\vOmega^v_i = \sigma^2_{v,i}\vC_v(\vtheta_{v,i})$ where $\vC_u(\vtheta_{u,i})$ and $\vC_v(\vtheta_{v,i})$ are valid correlation matrices (i.e., the matrices are positive definite; see Section~\ref{sec:PriorStructure}) and $\sigma_{u,i}^2$ and $\sigma_{v,i}^2$ are variance parameters.
For $\sigma^2, \sigma_{u,i}^2$ and $\sigma_{v,i}^2$ we assign the non-informative half-t prior on the standard deviation as proposed by \citet{Huang2013}; specifically $\sigma \sim \text{\textit{Half-t}}(1, A)$, $\sigma_{u,i} \sim \text{\textit{Half-t}}(1, A_{u,i})$ and $\sigma_{v,i} \sim \text{\textit{Half-t}}(1, A_{v,i})$.

One major benefit of our proposed prior is now realized: the full conditional distribution of $\widetilde{\vu}_i$ and $\widetilde{\vv}_i$ is proportional to a normal distribution (see supplement~\ref{sec:FCD}).
This results in a Gibbs update step for both $\widetilde{\vu}_i$ and $\widetilde{\vv}_i$ within the larger Markov chain Monte Carlo (MCMC) sampling scheme (shown in supplement~\ref{sec:FCD}), with computational benefits coming from known tricks for sampling from the normal distribution (e.g., the Cholesky decomposition).
Additionally, we have the ability to specify, or learn, unique correlation matrices $\vC_u(\vtheta_{u,i})$ and $\vC_v(\vtheta_{v,i})$ for each basis function which, to the best of our knowledge, has not been previously considered.

%%%%%%%%%%%%%%%%%%%%%%%%%%%%%%%%%%%%%%%%%%%%%%%%%%%%%%%%%%%%%%%%%%%%%%
\subsubsection{Special cases}\label{sec:specialCases}
%%%%%%%%%%%%%%%%%%%%%%%%%%%%%%%%%%%%%%%%%%%%%%%%%%%%%%%%%%%%%%%%%%%%%%

As discussed in Section~\ref{sec:PriorStructure}, when $\vOmega_i \equiv \vI$ our specified probabilistic model for SVD is equivalent to the fixed-rank SVD model proposed by \citet{Hoff2007}.
Another interesting property is the relationship to the classic algorithmic approach, C-SVD. 
As discussed and shown empirically through simulation in the supplement (\ref{sec:ASVD}), when $\vOmega_i = \vI$ the mean of the full conditional distribution for the basis functions is equivalent to the estimates obtained by C-SVD.

%%%%%%%%%%%%%%%%%%%%%%%%%%%%%%%%%%%%%%%%%%%%%%%%%%%%%%%%%%%%%%%%%%%%%%
\subsubsection{Model implementation}\label{sec:ModelSpecifications}
%%%%%%%%%%%%%%%%%%%%%%%%%%%%%%%%%%%%%%%%%%%%%%%%%%%%%%%%%%%%%%%%%%%%%%
% \subsection{Model parameters}
The SVD model (\ref{eqn:SVD_data}) has several parameters that need to be specified: the number of basis functions $k$, the correlation matrices $\vC_u(\vtheta_{u,i})$ and $\vC_v(\vtheta_{v,i})$, and any hyperparameters associated with the correlation matrices $\vtheta_{u,i}$ and $\vtheta_{v,i}$.
While in principle the value $k$ can be estimated either informally, e.g., scree plots \citep{Cattell1966}, or formally, e.g., cross-validation \citep{Wold1978} or the variable-rank model proposed by \citet{Hoff2007}, that is not the focus of this work.
Through empirical testing, we have found that if the true $k^*$ is less than the specified $k$, then the last $k-k^*$ basis functions of both $\vU$ and $\vV$ will have posterior credible intervals that cover zero at all, or nearly all, observations implying the basis function is not significant.
Conversely, if the true $k^*$ is greater than the specified $k$, there is little to no impact on the first $k$ basis functions (i.e., the $k$th basis function is not biased to account for the lost information by not estimating the remaining $k^* - k$ basis functions).
In choosing $k$ for the proposed model, an empirical Bayes approach could also be taken.
Specifically, one could compute the C-SVD, compute the cumulative amount of variance explained by the basis functions, and inform the value of $k$ based on this ``traditional'' approach.

Regarding the correlation matrices $\vC_{u,i}$ and $\vC_{v,i}$, as previously mentioned the hyperparameters $\vtheta_{u,i}$ and $\vtheta_{v,i}$ can either be specified directly or learned within the broader hierarchical model.
The latter choice would involve specifying a prior $p(\vtheta_{u,i}, \vtheta_{v,i})$ for these quantities and subsequently updating them within the MCMC algorithm.
In the case of using the Mat\'ern kernel to specify $\vC_u(\vtheta_{u,i})$ and $\vC_v(\vtheta_{v,i})$, recall that $\vtheta_{u,i} = \{\nu_{u,i}, \rho_{u,i}\}$ and $\vtheta_{v,i} = \{\nu_{v,i}, \rho_{v,i}\}$, where $\nu_{(\cdot)}$ describes the differentiability of the implied stochastic process and $\rho_{(\cdot)}$ describes the length-scale of the basis functions.
We generally recommend setting $\nu_{(\cdot)} = 3.5$ so the basis functions are third-order continuous but not over- or under- smoothed (e.g., infinitely differentiable with $\nu = \infty$ or non-differentiable with $\nu = 0.5$, respectively).
If the length-scale parameters are not estimated within the MCMC algorithm, they could be estimated offline via geostatistical techniques, e.g., estimating a semivariogram separately across both the rows and columns.
In the simulations presented in Section~\ref{sec:dataExamples} and for the application in Section~\ref{sec:SAT} we opt to estimate the length-scale parameters within the MCMC algorithm.

%%%%%%%%%%%%%%%%%%%%%%%%%%%%%%%%%%%%%%%%%%%%%%%%%%%%%%%%%%%%%%%%%%%%%%
\subsection{Other modeling choices}\label{sec:submodels}
%%%%%%%%%%%%%%%%%%%%%%%%%%%%%%%%%%%%%%%%%%%%%%%%%%%%%%%%%%%%%%%%%%%%%%

Section~\ref{sec:PSVD} proposes a general model for observed data using a low-rank approach.
However, there are other model specifications and corresponding matrix factorizations 
that can be seen as special cases of the SVD model. 
We discuss a few of these choices.

%%%%%%%%%%%%%%%%%%%%%%%%%%%%%%%%%%%%%%%%%%%%%%%%%%%%%%%%%%%%%%%%%%%%%%
\subsubsection{Principal components}\label{sec:PC}

As discussed in the introduction, PCA and SVD can be shown to produce an equivalent matrix factorization.
To this end, we can analogously represent the process $\vY = \vU \vA$ where $\vU$ is an orthonormal matrix of the eigenvectors of $\vY \vY'$, also known as the \textit{principal components}, $\vA = \vD\vV' = [\va_1, \ldots, \va_k]$ where $\va_i \sim N(0, \lambda_i\vI_m)$, and $\vLambda = diag(\lambda_1, \ldots, \lambda_k)$ are the eigenvalues of $\vY \vY'$, also known as the \textit{principal loadings}.
To estimate $\vU, \vA,$ and $\vLambda$ under this parameterization, there are two choices: (1) factor $\vE_{-i} = \vZ - \vU_{-i}\vA_{-i}$ in (\ref{eqn:SVDfactored}) and we assign the prior $\lambda_{i}\widetilde{\vu}_i|\vU_{i} \sim \mbox{N}_{n-k+1}(\vec{0}, \vN^{u \trp}_{i}\vOmega^u_i\vN^u_{i})\lambda_{i}^{n-1}$, or (2) estimate the parameters from the SVD model and compute $\vA$ as the posterior product of $\vD$ and $\vV'$.
For choice (1), only the columns of $\vU$ are dependent where the elements of $\vA$ are independent, resulting in only the principal components having dependence.
If choice (2) is taken, then the columns of $\vU$ and rows $\vA$ can be modeled dependently, where $\vA$ is dependent through the specification of $\vV$.
For PCA parameterization, we advocate for choice (2) as there is more control over the model than choice (1).

%%%%%%%%%%%%%%%%%%%%%%%%%%%%%%%%%%%%%%%%%%%%%%%%%%%%%%%%%%%%%%%%%%%%%%
\subsubsection{Including covariates}\label{sec:LinearTrend}

The general model (\ref{eqn:data}) allows for more complex model structure, such as including covariates.
Traditionally, data are centered, or de-trended, prior to computing the SVD/PCA decomposition.
However, within (\ref{eqn:data}) a mean term can be accommodated by modeling $\vM$.
We first consider a linear model for $\vM$, $\text{vec}(\vM) = \vX \vbeta$, where $\vX \in \mathbb{R}^{nm \times p}$ is a matrix of observed covariates and $\vbeta \in \mathbb{R}^{p}$ is a vector of unknown parameters.
To estimate $\vU, \vD, \vV$ under this parameterization, $\vE_{-i} = \vZ - [\vX \vbeta] - \vU_{-i}\vD_{-i}\vV'_{-i}$ in (\ref{eqn:SVDfactored}), where $[\vX \vbeta]$ denotes the reconstructed matrix of size $n \times m$.
To estimate $\vbeta$, we vectorize the model to get $\text{vec}(\vZ) \sim MVN_{nm}(\vX \vbeta + \text{vec}(\vU\vD\vV'), \vI_m \otimes \vSigma)$, assign the diffuse normal prior $\vbeta \sim MVN_p(\vec{0}, \sigma^2_{\beta}\vI_p)$, with $\sigma^2_{\beta}$ large, and get a standard normal-normal conjugate update for $\vbeta$.

This idea can be extended to a nonlinear function, say $\text{vec}(\vM) = f(\vX, \vbeta)$, where $f()$ is a nonlinear function.
For example, generalized additive models \citep{Hastie2017} or differential equations \citep{Berliner1996, Wikle2003a} could be used to model the nonlinear function.
However, care will likely need to be taken for the nonlinear case such that the nonlinear function is not too flexible, 
thereby conflicting with the random effect (e.g., see \ref{sec:covs}).

%%%%%%%%%%%%%%%%%%%%%%%%%%%%%%%%%%%%%%%%%%%%%%%%%%%%%%%%%%%%%%%%%%%%%%
\section{Synthetic data examples}\label{sec:dataExamples}
%%%%%%%%%%%%%%%%%%%%%%%%%%%%%%%%%%%%%%%%%%%%%%%%%%%%%%%%%%%%%%%%%%%%%%

We conduct three simulation studies to illustrate various aspects of the prior.
The first simulation provides justification for basis function-specific structure as opposed to a shared structure for all the basis functions.
The second illustrates how measurement error and model rank impact basis function recovery.
The last simulation investigates the ability to recover covariates when there may be confounding between the fixed and random effects.

%%%%%%%%%%%%%%%%%%%%%%%%%%%%%%%%%%%%%%%%%%%%%%%%%%%%%%%%%%%%%%%%%%%%%%
\subsection{Data generation}\label{sec:dataGeneration}

For all simulations, the target ``true'' basis functions $\vU$ and $\vV$ are simulated according to the generating mechanism described in section~\ref{sec:GeneratingMechanism} (e.g., to produce the orthonormal matrix in (\ref{eqn:orthomatrix})) with $\vOmega^u_i = \vC_u(\vtheta_{u,i})$ and $\vOmega^v_i = \vC_v(\vtheta_{v,i})$ where the elements of $\vC_u(\vtheta_{u,i})$ and $\vC_v(\vtheta_{v,i})$ are defined by the Mat\'ern correlation function with $\vtheta_{u,i} = (\nu_{u,i}, \rho_{u,i})$ and $\vtheta_{v,i} = (\nu_{v,i}, \rho_{v,i})$.
Data is simulated according to $Z(x,t) \sim N(M(x,t) + Y(x,t), \sigma^2)$ with $x = x_1, \ldots, x_n$ equally spaced in $\mathcal{X} = [-5, 5]$, $t = t_1, \ldots, t_m$ equally spaced in $\mathcal{T} = [0, 10]$, $n = 100$, $m = 100$, and $Y(x,t)$ being the $(x,t)$ element of the matrix $\vY = \vU \vD \vV'$.
The specification of $\vM=[M(x,t)]_{(x,t)\in\mathcal{X}\times\mathcal{T}}$ is described in each of the following subsections.
The value of $\sigma^2$ is chosen to match a target signal-to-noise ratio (SNR): let $\veta$ be a random $n \times m$ matrix of iid standard normal random variables, then, $\sigma = \sqrt{\frac{var(\vM + \vY)}{\text{SNR} * var(\veta)}}$. Ultimately, the simulated data is $\vZ = \vM + \vY + \sigma \veta$ (see the supplement Figure~\ref{fig:simdata1D}) with $var(\vM + \vY)/var(\vZ-\vM-\vY) = \text{SNR}$.

%%%%%%%%%%%%%%%%%%%%%%%%%%%%%%%%%%%%%%%%%%%%%%%%%%%%%%%%%%%%%%%%%%%%%%
\subsection{Synthetic example \#1: basis function-specific length scales}\label{sec:variablelength}
%%%%%%%%%%%%%%%%%%%%%%%%%%%%%%%%%%%%%%%%%%%%%%%%%%%%%%%%%%%%%%%%%%%%%%

The first simulation study assesses how our model recovers the underlying basis functions when the true basis functions have differing length-scales.
We compare our ``variable model,'' in which we allow each basis function to have unique structure that is estimated from the data, to a ``grouped model,'' in which all basis functions have a shared structure that is also estimated from the data.
A distinguishing feature of our methodology is that we can model basis function-specific structure, in comparison to other recent work \citep{Pourzanjani2021, Jauch2021} wherein all basis functions have the same length-scales. 
Both models are described in Section \ref{sec:PSVD}: in the variable model, $\rho_{\cdot,i}$ and $\rho_{\cdot,j}$ need not be equal, while in the grouped model, we impose the restriction that $\rho_{\cdot,i} = \rho_{\cdot,j}$, for $i ,j = 1, \ldots, k$.
The grouped model is a special case of the variable model, illustrating the enhanced flexibility of our methodology relative to existing approaches.

To explore the effect of basis function-specific structure, we generate data where the length-scale for each basis function varies from larger to smaller in an exponentially decreasing trend similar to what is shown in Figure \ref{fig:lengthScaleMotivation}.
To determine the effect of measurement error in conjunction with varying basis function length-scale, we generate data sets with $\sigma^2$ chosen such that 
$\text{SNR} = [10, 5, 2, 1, 0.5, 0.1]$.
For this simulation study, we do not consider the effect of $\vM$, and all data is simulated with $\vM \equiv 0$.
For all data generation, we specify the true number of basis functions $k = 4$ with covariance parameters $\nu_{(\cdot),i} = 3.5$ for $i=1,\dots,k$ and $\vrho_{(\cdot)} = (3.5, 1, 0.5, 0.25)$ for both $\vU$ and $\vV$, and diagonal matrix $\vD = \mbox{diag}(40, 30, 20, 10)$.
For each SNR, we obtain 10000 posterior samples of the model parameters and discard the first 5000 as burn-in for both the variable and grouped model.
The process is repeated 100 times for each SNR to help understand the variability in the results.

\begin{figure}[!t]
    \centering
    \includegraphics[width = \linewidth]{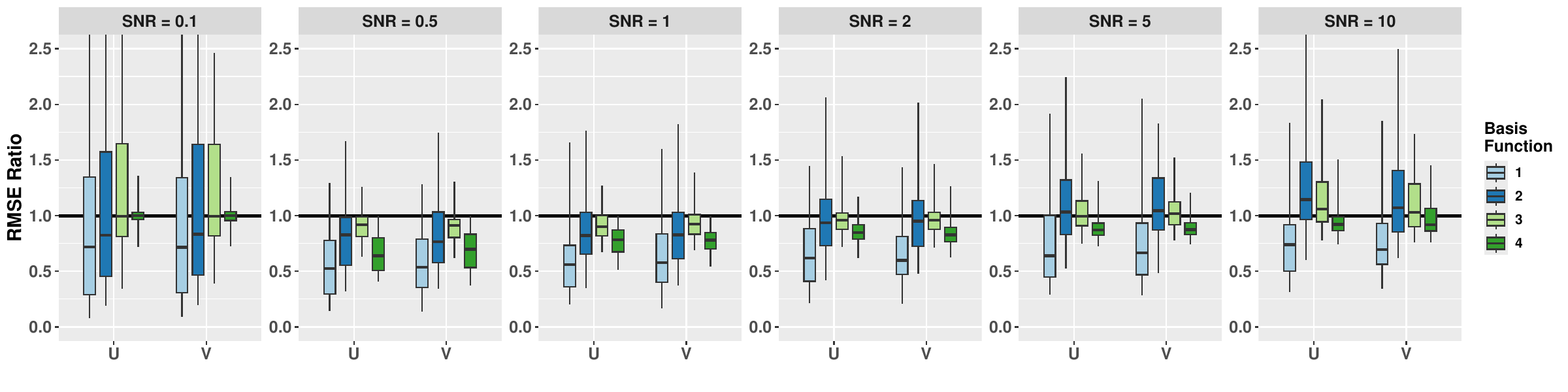}
    \caption{Box plots of the ratio of the RMSE for the variable model divided by the grouped model for $\vU$ (left) and $\vV$ (right) stratified by SNR (sub-panel) and basis function (color) with a horizontal line at 1. For each box, the lower and upper hinge are the 25th and 75th percentiles, respectively, the line within the box is the median, and the lower and upper whiskers are 2.5 and 97.5 percentiles. Note, we have limited the y-axis to ease visual comparison between panels and only the first panel, with a SNR = 0.1, has values outside of the range.}
    \label{fig:VariableSimulationResults}
\end{figure}

For each simulation and model, we calculate the element-wise average root mean squared error (RMSE) of the posterior mean for each basis function in $\vU$ and $\vV$ compared to their corresponding true value.
To compare the RMSE estimates of the variable to grouped model, 
Figure~\ref{fig:VariableSimulationResults} shows the ratio of the RMSE estimate for the variable model over the group model for $\vU$ (left) and $\vV$ (right) stratified by the SNR (sub-panels) and by the basis function (color) along with a horizontal reference line at one.

RMSE ratios less than one favor the variable model.
From the figure, we see basis functions 2 and 3 for both $\vU$ and $\vV$ have ratios closest to 1 for all values of SNR.
In contrast, basis functions 1 and 4 for both $\vU$ and $\vV$ have ratios that are systematically less than 1 for all values of SNR except 0.1.
The reason the variable model has improved RMSE performance for 1 and 4 is because the estimate for $\rho$ for the grouped model is pulled toward the average length-scale value, which is close to the true length scale for basis functions 2 and 3.
This bias results in the grouped model over-fitting basis function 1 (since the pooled estimate of the length scale is less than the true length scale) and under-fitting basis function 4 (since the pooled estimate of the length scale is larger than the true length scale); see estimates in Figure \ref{fig:variableLengthPlot} for a visual example of the over- and under-fitting.

In summary, our first synthetic example verifies that when the data have differing structures in the underlying basis functions, failing to account for those different structures results in systematically larger errors in the basis function estimates. The true structures can only be appropriately captured when the underlying statistical model directly accounts for basis function-specific structure.

%%%%%%%%%%%%%%%%%%%%%%%%%%%%%%%%%%%%%%%%%%%%%%%%%%%%%%%%%%%%%%%%%%%%%%
\subsection{Synthetic example \#2: model rank}\label{sec:rankSimulation}
%%%%%%%%%%%%%%%%%%%%%%%%%%%%%%%%%%%%%%%%%%%%%%%%%%%%%%%%%%%%%%%%%%%%%%

We now conduct a simulation study to illustrate the impact of SNR and model rank $k$ on basis function recovery.
To determine the effect of measurement error, we again generate data sets with $\sigma^2$ chosen such that $\text{SNR} = [10, 5, 2, 1, 0.5, 0.1]$.
As with the previous simulation study, all data is simulated with $\vM \equiv 0$.
For all data generation, we set the true number of basis functions $k^* = 5$ with covariance parameters $(\nu_{(\cdot),i}, \rho_{(\cdot),i}) = (3.5, 3)$ for both $\vU$ and $\vV$ and for all $i=1,\dots,k^*$, and diagonal matrix $\vD = \mbox{diag}(40, 30, 20, 10, 5)$.
One realization of the simulated data with  $\mbox{SNR}=1$ and the $\vU$ and $\vV$ basis functions are shown in Figure~\ref{fig:data1D} in the supplement.

As discussed in Section~\ref{sec:ModelSpecifications}, using this model only requires specification of $k$, the number of basis functions used in $\vU$ and $\vV$, and kernels for $\vC_{u}(\vtheta)$ and $\vC_{v}(\vtheta)$.
To investigate how possible mis-specification of the number of basis functions impacts model recovery, we estimate the model with $k = [3, 4, 5, 6, 7]$ for each level of SNR.
Additionally, we specify a Mat\'ern kernel with smoothness parameter $\rho = 3$ for the correlation structure for all basis functions.
For each SNR and $k$ combination, we obtain $10000$ posterior samples of the model parameters, discarding the first $5000$ as burn-in.
We repeat this process 100 times.

For each posterior simulation, we calculate the 95\% coverage rate (CR) and RMSE for $\vU$, $\vV$, and the ``true'' surface $\vY = \vU \vD \vV'$.
If the true $k^*$ is greater than the specified $k$, the empirical CR and RMSE are computed only for the first $k$ basis functions and then averaged over the $k$ estimates (e.g., we do not consider the last $k^* - k$ basis functions when computing CR and RMSE).
If the true $k^*$ is less than the specified $k$, the empirical CR and RMSE for the ``extra'' $k-k^*$ basis functions are compared to the zero line and the reported CR and RMSE values are obtained by averaging over the $k$ estimates.
Additionally, for each simulation we computed the C-SVD using the base linear algebra library, \textit{LinearAlgebra.jl}, in Julia \citep[][]{Bezanson2017} and computed the RMSE of the calculated $\vU, \vV$, and reconstructed surface $\vY$ assuming the same truncation value $k$.
The coverage rates and the RMSE are shown in Figure~\ref{fig:simulation_CR_RMSE}.
The results of one simulation are shown in Figure~\ref{fig:estimatedBasis1D} based on the data shown in Figure~\ref{fig:data1D} in the supplement.

\begin{figure}
    (a) Coverage rate, aggregated across repeated samples
    
    \includegraphics[width = 0.95\textwidth]{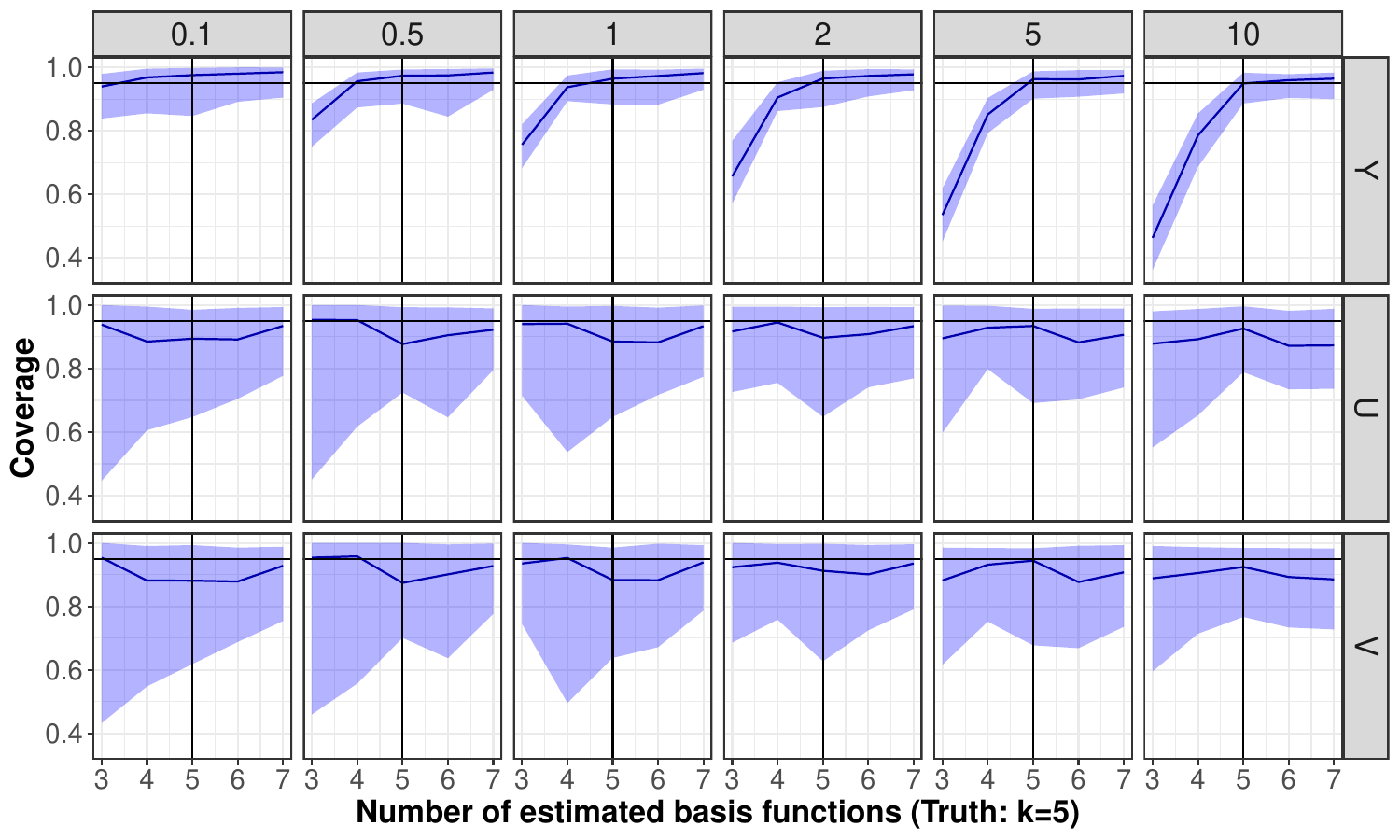}
    
    (b) Root mean square error, aggregated across repeated samples
    
    \includegraphics[width = 0.95\textwidth]{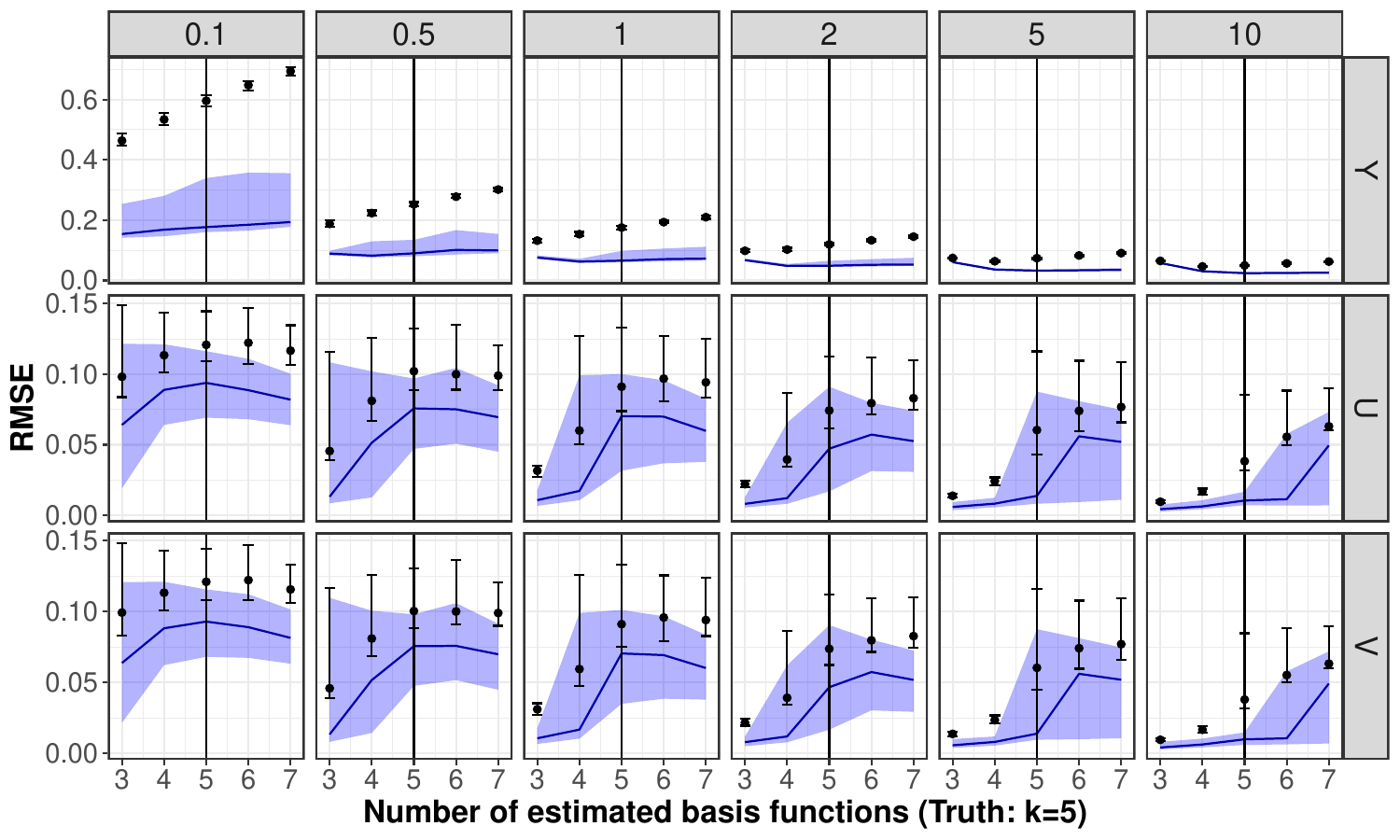}

    \caption{\small{Validation results from the synthetic data example, showing coverage rate (top) and root mean square error (bottom). In each panel, the solid blue line is the median Monte Carlo coverage rate and shaded regions are the 95\% Monte Carlo uncertainty bounds for the coverage rate over synthetic replicates. Results are shown for varying levels of SNR and values of $k$ for the recovered surface $\vY$ (top), $\vU$ basis functions (middle), and $\vV$ basis functions (bottom). The SNR values range from 0.1 (left) to 10 (right). The black vertical line indicates the true value $k^*=5$  and the horizontal black line for (a) is at 95\% (the nominal coverage rate). In panel (b), the black point and error bars show the median and 95\% bootstrapped confidence interval for the RMSE using the algorithmic C-SVD method.}}
    \label{fig:simulation_CR_RMSE}
\end{figure}

\begin{figure}[t]
    \centering
    \includegraphics[width = \textwidth]{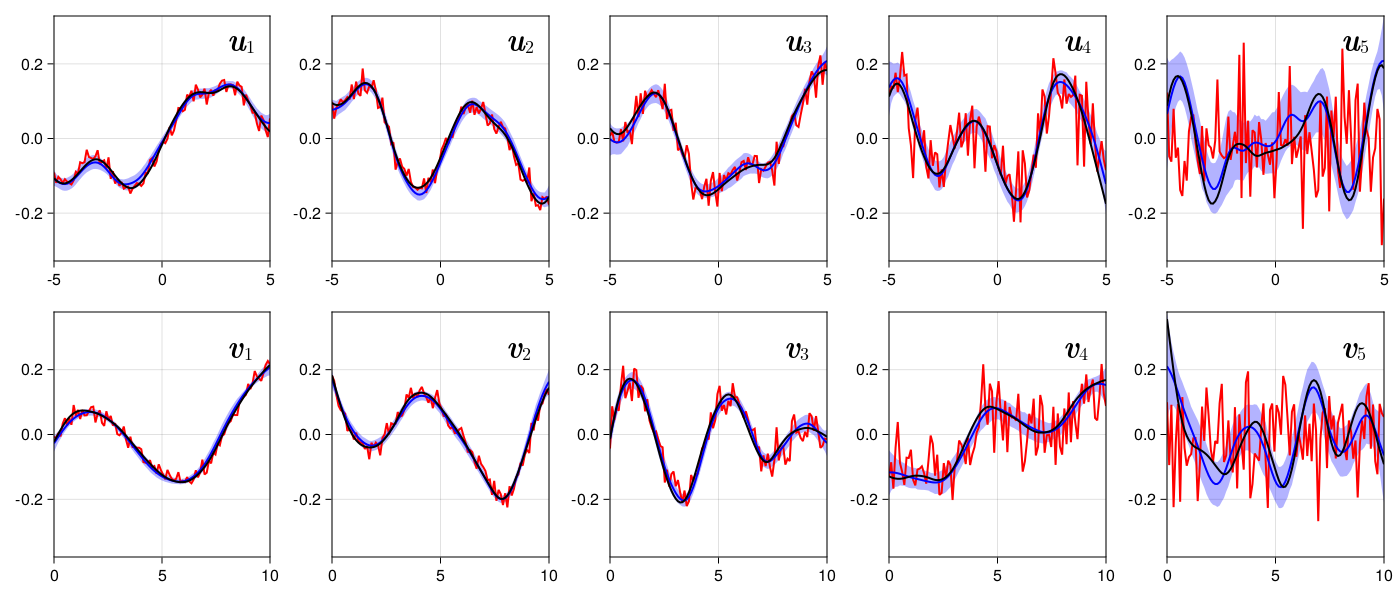}
    \caption{Posterior mean (blue line), 95\% credible intervals (shaded blue region), truth (black line), and C-SVD estimate (red line) for the $\vU$ and $\vV$ basis functions from a random simulation. The data associated with this random simulation is shown in Figure~\ref{fig:data1D}.}
    \label{fig:estimatedBasis1D}
\end{figure}

From Figure~\ref{fig:simulation_CR_RMSE}(a), we see our median coverage rate for the $\vU$ (middle row) and $\vV$ (bottom row) basis functions (blue line) is near the nominal level (horizontal black line) and the 95\% Monte Carlo uncertainty bounds (MCUB) for the coverage rate (blue shaded region) covers the nominal level for all SNR levels and regardless of the specification of $k$.
This implies that posterior uncertainties are well calibrated and robust to mis-specifications of the number of estimated basis functions, regardless of the magnitude of the noise.
For the recovered data (top row), we see the 95\% MCUB cover the nominal level for all SNR levels with $k$ greater than 5.
However, for $k$ less than 5, achieving the nominal coverage depends on SNR: in low signal cases (e.g., SNR = 0.1), the uncertainties are well calibrated, while posterior uncertainties are too small (i.e., coverage of the truth is much less than the nominal level) when the signal is stronger (SNR $>$ 0.5).
This counterintuitive result is due to the impact of unaccounted signal for higher-order basis functions ($i=4$ and/or $i=5$) on the signal: for large SNR, individual basis functions both (a) contribute more to the overall uncertainty in the data and also (b) have narrower posterior distributions, such that ignoring one or more true basis functions causes the model to underestimate data uncertainties (e.g., see Figure~\ref{fig:estimatedBasis1D}).
Conversely, for smaller SNR, there is more uncertainty in each basis function estimate and the impact of higher-order basis functions on the estimated surface is reduced, to the extent that the model can recover the nominal coverage of the data.

For the RMSE, shown in Figure~\ref{fig:simulation_CR_RMSE}(b), the most notable result is that the median RMSE for our approach (blue line) is systematically lower than the corresponding RMSE from the algorithmic C-SVD approach for both data (top row) and basis functions (middle and bottom rows), across SNR levels and specification of $k$.
In other words, estimates of the basis functions in both $\vU$ and $\vV$ and the recovered data have systematically lower errors than what one can obtain from the algorithmic approach. 
Regarding RMSE for estimates of the recovered surface, the median error (blue line) decreases as a function of SNR, as expected, and interestingly the data RMSE is relatively insensitive to specification of $k$. 
For the $\vU$ and $\vV$ basis functions (middle and bottom rows, respectively, of Figure~\ref{fig:simulation_CR_RMSE}), we see that trajectories of RMSE estimates for our proposed approach and the C-SVD mirror each other, with our estimates being systematically, but not significantly, lower.
However, across SNR levels, the RMSE actually increases as one moves from $k=3$ to $k=7$ (even though the true $k^*=5$).
For SNR equal to 5 and 10, we see a dramatic spike in the RMSE estimate and uncertainty for the $\vU$ and $\vV$ basis functions for $k = 6$ and $7$.
This is because we are comparing against the zero line for these cases: while the uncertainty bounds for these basis function covers the zero line (as seen in the coverage results in Figure~\ref{fig:simulation_CR_RMSE}a.), there is a lot of variability in these estimates (with relatively lower uncertainty due to larger signal), leading to inflated RMSE values.

In conclusion, this synthetic data example shows the proposed method has well calibrated uncertainty and significantly reduces the impact of measurement noise on the basis function estimates.
However, there is a significant trade-off in choosing $k$ to be too small or large based on the magnitude of the SNR.
Based on our simulation, there will be significant bias in the recovered surface but \textit{not} in the estimated basis functions if $k$ is too small and the SNR is low.
Additionally, there will \textit{not} be significant bias in the recovered surface or in the estimated basis functions if either $k$ is too small and the SNR is large or $k$ is too large.
The only trade-off for $k$ too large is inflated RMSE's for the extraneous basis functions, which could lead to underestimated RMSE's in the recovered surface.
Therefore, we suggest erring on the side of choosing $k$ to be too large.

%%%%%%%%%%%%%%%%%%%%%%%%%%%%%%%%%%%%%%%%%%%%%%%%%%%%%%%%%%%%%%%%%%%%%%
\subsection{Synthetic example \#3: covariates}\label{sec:covs}
%%%%%%%%%%%%%%%%%%%%%%%%%%%%%%%%%%%%%%%%%%%%%%%%%%%%%%%%%%%%%%%%%%%%%%

To illustrate how covariates impact the estimation of the basis functions, we now include the fixed effect $\vM$ when simulating data and specify the SNR to be 2.
We consider three different cases of the model for $\vM$: (\ref{M1}) independent fixed and random effects, (\ref{M2}) strongly confounded spatial and temporal fixed and random effects, and (\ref{M3}) weakly confounded spatial and temporal fixed and random effects.
For all three models, we specify $\text{vec}(\vM) = \vX \vbeta$ where $\vbeta = (\beta_1, \ldots, \beta_4) = (-2, 0.6, 1.2, -0.9)$ and $\vX$ is a $nm$ by $4$ matrix.
For each model, the covariates are generated as:
\begin{enumerate}[label=$\mathrm{M}\arabic*$]
    \item\label{M1} - Each element of $\vX$ is i.i.d. $N(0, 0.2^2)$.
    
    \item\label{M2} - Let $\widetilde{\vx}_{1,s}, \widetilde{\vx}_{2,s} \sim N_n(\vec{0}, \vSigma_{s})$, $\widetilde{\vx}_{t} \sim N_m(\vec{0}, \vSigma_{t})$, and $\vx_{st} \sim \mbox{N}_{nm}(\vec{0}, \vSigma_{st})$ where $\vSigma_{s}, \vSigma_{t}$, and $\vSigma_{st}$ are correlation matrices specified using the Mat\'ern kernel with smoothness parameter $\nu = 3.5$ and length-scale parameter $\rho = 3,3$ and $1$, respectively, which is equal to the length-scale of the spatial and temporal random effect, respectively. Then, $\vX = [\vx_{1,s}, \vx_{2,s}, \vx_{t}, \vx_{st}]$ is a $nm \times 4$ matrix where $\vx_{1,s} = \vI_m \otimes \widetilde{\vx}_{1,s}, \vx_{2,s} = \vI_m \otimes \widetilde{\vx}_{2,s}$, and $\vx_{t} = \widetilde{\vx}_{t} \otimes \vI_n$.
    
    \item\label{M3} - The covariate matrix is created in the same manner as in M2 except the length-scale of $\vSigma_{s}, \vSigma_{t}$, and $\vSigma_{st}$ are $\rho = 0.3, 0.3$ and $1$, respectively.
\end{enumerate}

For each covariate specification M1--M3, we implement our methodology with $k=5$, a Mat\'ern kernel with smoothness parameter $\nu = 3$ for the correlation structure for all basis functions, and a diffuse normal prior, $\mbox{N}(0, 10^2)$, on each element of $\vbeta$.
We obtain 10000 posterior samples of the model parameters, discarding the first 5000 as burn-in.
Posterior summaries of the regression coefficients are shown for each model in Table~\ref{tab:indepbeta}.
From the table, we see only $\beta_3$ from \ref{M2} has a credible interval that does not cover the true value, indicating the model is able to reasonably recover the fixed effects under all three scenarios.
To determine the model's ability to correctly recover the random effect, we computed the point-wise 95\% posterior coverage rate for the random effect $\vY = \vU \vD \vV'$ for each \ref{M1}, \ref{M2}, and \ref{M3}, which are 0.965, 0.276, and 0.984, respectively.
Therefore, when the fixed and random effect are independent or they have different spatial and temporal frequencies (weakly confounded), the model is able to correctly identify both model components.
When the fixed and random effects have similar, or in this example equal, spatial and temporal frequencies, the model is unable to properly capture the random effect but can still capture the fixed effect.

\begin{table}[!t]
    \centering
    \begin{tabular}{c|l||cccc}
         model & & $\beta_1$ & $\beta_2$ & $\beta_3$ & $\beta_4$ \\
        \hline
               & true     & -2     & 0.6   & 1.2   & -0.9   \\
        \hline
                  & mean     & -2.032 & 0.632 & 1.204 & -0.873 \\
         \ref{M1} & lower CI & -2.082 & 0.583 & 1.156 & -0.921 \\
                  & upper CI & -1.983 & 0.680 & 1.252 & -0.824 \\
         \hline
                  & mean     & -1.995 & 0.636 & 1.071 & -0.875 \\
         \ref{M2} & lower CI & -2.036 & 0.531 & 1.014 & -0.913 \\
                  & upper CI & -1.943 & 0.747 & 1.118 & -0.807 \\
        \hline
                  & mean     & -2.005 & 0.601 & 1.194 & -0.908 \\
         \ref{M3} & lower CI & -2.016 & 0.589 & 1.179 & -0.938 \\
                  & upper CI & -1.994 & 0.614 & 1.209 & -0.882 \\
    \end{tabular}
    \caption{Posterior mean (top row), lower 95\% credible interval (middle row), and upper 95\% credible interval (bottom row) for the regression coefficients of models M1--M3 (top-bottom).}
    \label{tab:indepbeta}
\end{table}

Based on previous work by \citet{Paciorek2010} discussing the issue of scale with spatial mixed-effects models, our results are not surprising.
Specifically, if the fixed and random effects operate on different scales (either spatially or temporally), \citet{Paciorek2010} rigorously argues the fixed and random effects are identifiable.
If they operate on similar (or equivalent) scales, they are not identifiable.
If interpretation of the random effect is not important, the random effect can restricted to be orthogonal to the fixed effect, thereby making the random effect identifiable on the space orthogonal to the fixed effect \citep{Reich2006, Hodges2010, Hanks2015a}.
However, there has been debate as to the validity of modeling the random effect on the restricted space \citep{Zimmerman2022}.
Because this is not the main goal of the paper, for now we simply recommend being cognizant of these issues.

%%%%%%%%%%%%%%%%%%%%%%%%%%%%%%%%%%%%%%%%%%%%%%%%%%%%%%%%%%%%%%%%%%%%%%
\section{Surface air temperature}\label{sec:SAT}
%%%%%%%%%%%%%%%%%%%%%%%%%%%%%%%%%%%%%%%%%%%%%%%%%%%%%%%%%%%%%%%%%%%%%%

As discussed in the introduction, empirical orthogonal functions, or EOFs, are commonly used in climate sciences to summarize modes of variability in atmospheric systems.
Typically, external factors that could be driving the system are referred to as \textit{climate forcings} and modeled as fixed effects, while ``unforced'' year-to-year variability is modeled as a spatial, temporal, or spatio-temporal random effect and referred to as \textit{internal variability}. 
Importantly, when EOF analysis is applied to climate data where the long term trends have been removed, this can be considered a method for characterizing the internal variability of the system.
Particularly for extreme temperature events, EOFs are an important tool for understanding how internal variability combines with long-term trends to produce short-term events that have a large impact on human systems \citep{Grotjahn2016}.
Historically, estimates of the internal variability are derived from ensembles of climate models and rarely computed from observational data products.
Here, we explore our ability to estimate the internal variability of monthly maximum two-meter air temperature in the Pacific Northwest, where it is important to account for spatial and temporal structures in the extreme measurements \citep[again see, e.g.,][]{Grotjahn2016}. 
Such estimates are important for understanding the statistics of monthly maximum temperatures in this region, particularly in light of the recent devastating heatwave that impacted this region in the summer of 2021 \citep{BercosHickey2022}.

We use gridded monthly maximum two-meter air temperature data (tXx) by extracting the largest daily maximum two-meter air temperature each month from the ERA5 reanalysis dataset \citep{Hersbach2020} at 0.25$^\circ$ horizontal resolution from January 1979 to December 2021.
The data are centered by subtracting off the global mean.  
We focus on the subset of data from 44$^\circ$- 53$^\circ$N and 116$^\circ$- 128$^\circ$W, for a total of 1813 spatial locations across 516 time points.
While it is possible to include relevant covariates for this analysis (e.g., greenhouse gas emissions, the El Ni\~no/Southern Oscillation, urbanization, and drought conditions) using a model for $\vM$ (e.g. Section \ref{sec:LinearTrend}), this would have resulted in a substantial number of parameters to estimate and is not the main focus of this work.
Therefore, we opt instead to focus on the model for the random effect and simply centered the data \textit{a priori} to parameter estimation.

As discussed in the introduction, Figure~\ref{fig:lengthScaleMotivation} shows empirical evidence that the basis functions resulting from a SVD of tXx may have different structure. 
We proceed with this assumption.
Therefore, we parameterize the covariance matrix for the prior of the spatial basis functions using the Mat\'ern kernel with smoothness $\nu = 3.5$ and the covariance matrix for the prior of the temporal basis functions using the Gaussian kernel. 
The effective range for both the spatial and temporal basis functions are estimated along with other model parameters.
We specify $k=10$ based on the first 10 basis functions explaining approximately 99\% of the variance as determined from the C-SVD decomposition.
We obtain $10000$ samples from the posterior, discarding the first $5000$ as burn-in, where convergence is assessed graphically with no issues detected.

Posterior summaries of three spatial basis functions (2, 5, and 7), three temporal basis functions (2, 5, and 7), and all length-scale estimates are shown in Figure~\ref{fig:T2Msummary} A), B), and C), respectively.
We highlight basis function 2 because it has little to no significant difference between C-SVD estimate, and 5 and 7 because they contain many spatial and temporal locations with significant differences.
Panel a) depicts summaries of three spatial basis functions $\vu_2$ (top), $\vu_5$ (middle), and $\vu_7$ (bottom), where the left column are the estimates from C-SVD, the middle column are the posterior means from our proposed model, and the right column are the posterior difference between the posterior mean and the algorithmic estimate where locations whose 95\% credible interval does not cover zero are denoted with an `x'.
Panel b) contain estimates of three temporal basis functions $\vu_2$ (top), $\vu_5$ (middle), and $\vu_7$ (bottom), where the black line is the C-SVD estimates, blue line is the posterior mean from our proposed model, and blue shaded region are the 95\% CIs where a vertical line denotes the 95\% CI does not cover the C-SVD estimate.
The last panel, c), are posterior mean estimates of the length-scale parameter (dot) and 95\% credible intervals (error bars) of the correlation kernel for each spatial (left) and temporal (right) basis functions, where blue estimates correspond to the selected basis functions for panels a) and b).
Posterior summaries of all 10 spatial and temporal basis functions are included in the supplement.

\begin{figure}[!t]
    \centering
    \includegraphics[width = \linewidth]{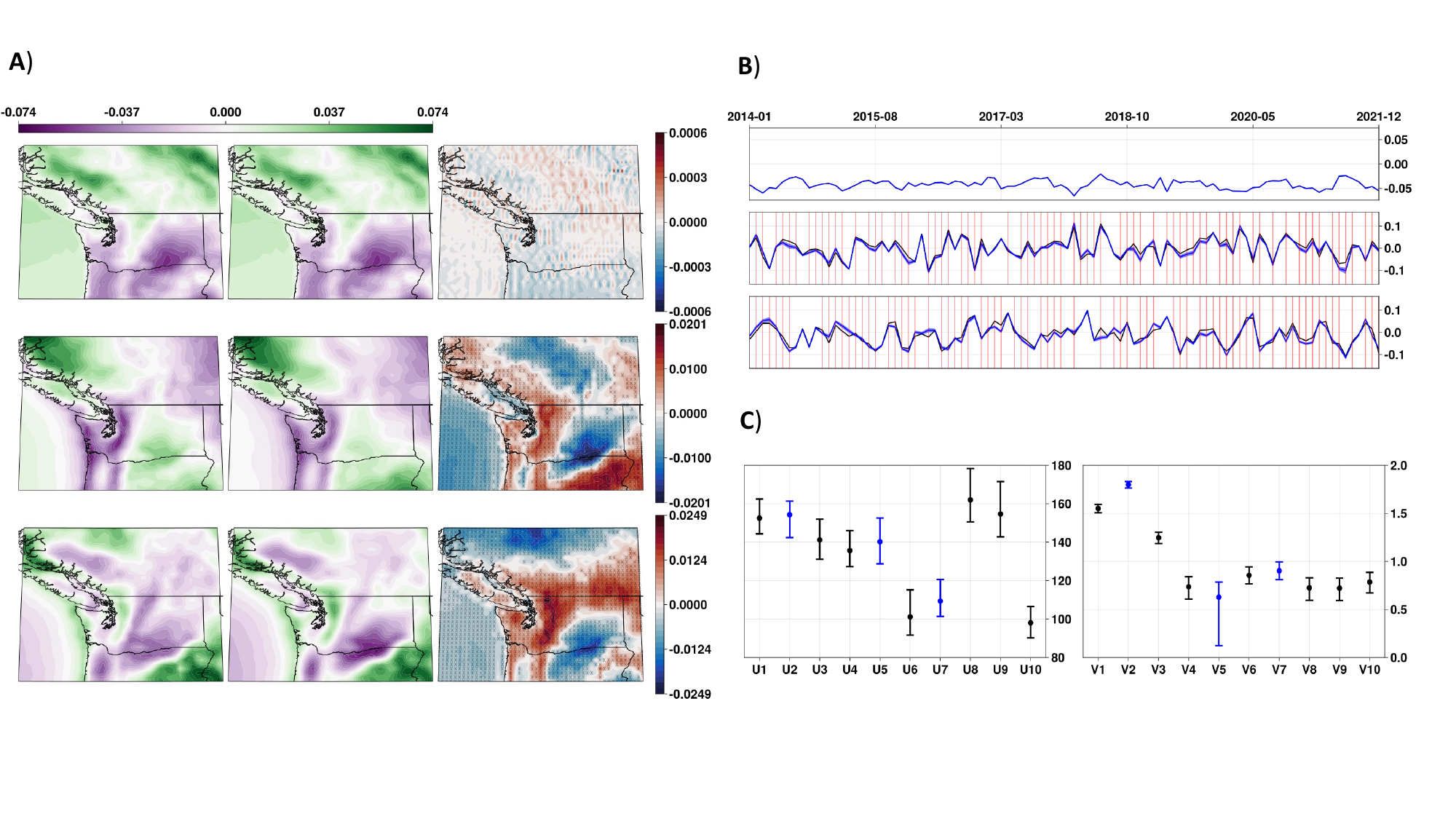}
    \caption{\textbf{a)} Estimated spatial basis functions $\vu_2$ (top), $\vu_5$ (middle), and $\vu_7$ (bottom). The left column are the estimates from C-SVD, the middle column are the posterior means, and the right column are the posterior difference between the posterior mean and the algorithmic estimate where locations whose 95\% credible interval does not cover zero are denoted with an `x'.
    \textbf{b)} Estimated temporal basis functions $\vv_2$ (top), $\vv_5$ (middle), and $\vv_7$ (bottom). For each panel, the black line are the estimates from C-SVD, blue line are the posterior means, and blue shaded region are the 95\% CIs where a vertical line denotes the 95\% CI does not cover the C-SVD estimate.
    \textbf{c)} Posterior mean estimate of the length-scale parameter (dot) and 95\% credible intervals (error bars) of the correlation matrix for each spatial (left) and temporal (right) basis functions. Estimates in blue correspond to the selected basis functions for panels \textbf{a)} and \textbf{b)}.}
    \label{fig:T2Msummary}
\end{figure}

Comparing the spatial plots of the posterior mean to the deterministic counterpart (Figure \ref{fig:T2Msummary}A), the posterior estimates are much smoother spatially and for the fifth and seventh basis functions, the estimates are significantly different over much of the spatial region.
The estimates, both deterministic and probabilistic alike, have an interpretation that makes sense physically.
The second basis function (top row) has a clear land-sea contrast and distinguishes between the plains (purple) and mountains (green).
The fifth basis function captures the influence of the low-lying coastal region and foothills of Canadian Rockies (purple) in contrast to the wet/dry regimes in Canada and Oregon/Washington (green).
The seventh combines multiple physical features and aligns well with geographical features such as topography and appears to capture steep gradient contrasts.

Regarding the temporal estimates (Figure \ref{fig:T2Msummary}B), the second basis function (top) does not have any time points with significantly different estimates than the C-SVD counterpart.
However, both the fifth (middle) and seventh (bottom) do have significant differences (denoted with the vertical lines), and we see the posterior means produce smoother estimates than the C-SVD counterparts.

Additionally, the basis functions all have different posterior mean length-scale estimates.
For the spatial basis functions, $\vu_6, \vu_7,$ and $\vu_{10}$ have significantly smaller values than the other six, as determined by the range of the 95\% CI (Figure \ref{fig:T2Msummary}C, left), and for the temporal, the first three have significantly larger values than the other seven, as determined by the range of the 95\% CI (Figure \ref{fig:T2Msummary}C, right).
This shows we are able to capture the spatial and temporal relationship within each basis function and that the spatial and temporal relationship is different across basis functions.

Importantly for climate science, we are able to provide estimates of the internal variability of a system from observational data, in this example monthly maximum two-meter air temperature of daily maxima, by reconstructing the internal variability using posterior estimates of our structured basis functions.
The estimates account for measurement uncertainty, spatial and temporal dependence, and have quantifiable uncertainty.
These estimates can then be used to account for the internal variability of a system and help isolate the extent to which external factors are driving changes to the system.
Additionally, producing ensembles of weather variables like extreme temperature using climate models can be computationally intensive.
However, we can now sample directly from the posterior distribution of the internal variability of extreme monthly temperatures, accounting for the spatial structures innate to the underlying data. 
These posterior samples are analogous to ensembles of the climate system and computationally much cheaper to compute than ensembles of climate model runs.

%%%%%%%%%%%%%%%%%%%%%%%%%%%%%%%%%%%%%%%%%%%%%%%%%%%%%%%%%%%%%%%%%%%%%%
\section{Discussion}\label{sec:discussion}
%%%%%%%%%%%%%%%%%%%%%%%%%%%%%%%%%%%%%%%%%%%%%%%%%%%%%%%%%%%%%%%%%%%%%%

We proposed a novel prior distribution for structured orthonormal matrices that is an extension of \citet{Hoff2007}, where the individual basis functions can be modeled dependently.
The prior is based on the projected normal distribution which we augment with a latent length parameter.
When our prior is combined with a normal data model, the resulting full conditional distributions for the basis functions are conjugate, resulting in analytically straightforward MCMC sampling.
We describe how the prior can be used to conduct posterior inference on a general class of probabilistic SVD models and how to extend the proposed model to various other applications.
We discussed various mathematical properties of our probabilistic SVD model (supplement~\ref{sec:IC}) and illustrated its capability through multiple simulation studies.
The model is then used to draw inference on the internal variability of extreme two-meter air temperature, allowing us to quantify space-time structures in a complex climate process.

The synthetic data examples and application presented in Sections~\ref{sec:dataExamples} and \ref{sec:SAT}, respectively, all highlight the model's efficacy on gridded, i.e., uniformly spaced, data.
However, the model is equally well suited for non-uniformly spaced data so long as the spacing is consistent within space and within time.
If the data are not spaced consistently within space and within time, this would constitute a missing data problem, which we plan to explore in future work.
In addition, our model assumes normally distributed errors.
This assumption can be relaxed by, for example, assuming a hierarchical structure and modeling the mixed-effects as a latent process.

Another area for possible extension could explore the concept of regularized basis functions through the posterior mode of the basis functions.
Similar to the Bayesian Lasso \citep{Park2008} or Bayesian Group Lasso spatial data \citep{Hefley2017}, an $\ell_1$ penalty could be imposed by representing a Laplace distribution as a scale mixture of normal distributions.
The addition of a penalty term, especially a penalty that forces values to zero, could produce sparse dependently structured basis functions whose importance within the spatial context is explored by \citet{Wang2017}.

The choice of the number of basis functions, $k$, is the only major subjective choice in our proposed probabilistic SVD model.
While we show the mis-specification of $k$ does not have a negative impact when erring on the side of $k$ being too large, a more flexible model estimating $k$ is attractive.
To estimate $k$, \citet{Hoff2007} proposed a variable-rank model utilizing the so-called spike-and-slab variable selection prior \citep{Mitchell1988}.
However, because of the difference in our prior compared to the prior proposed by \citet{Hoff2007}, incorporating the spike-and-slab prior into our proposed model would require extra theoretical work.
Work focused on estimating the rank $k$ with our framework would produce a very flexible approach for modeling spatio-temporal random effects.

Finally, our proposed prior does have the disadvantage of relying on a column-wise sampling strategy.
Specifically, within each MCMC iteration, there is a required $\mathcal{O}(n^3)$ cost of computing the orthonormal basis for the null-space $\vN_i^u$ and $\vN_i^v$ (see the supplement for more discussion).
The additional flexibility our approach offers comes at the cost of the computational gains from the methods by \citet{Pourzanjani2021} and \citet{Jauch2021}, which propose solutions to this column-wise strategy.

%%%%%%%%%%%%%%%%%%%%%%%%%%%%%%%%%%%%%%%%%%%%%%%%%%%%%%%%%%%%%%%%%%%%%%
\section{Competing interests}
%%%%%%%%%%%%%%%%%%%%%%%%%%%%%%%%%%%%%%%%%%%%%%%%%%%%%%%%%%%%%%%%%%%%%%
No competing interest is declared.

%%%%%%%%%%%%%%%%%%%%%%%%%%%%%%%%%%%%%%%%%%%%%%%%%%%%%%%%%%%%%%%%%%%%%%
\section{Author contributions statement}
%%%%%%%%%%%%%%%%%%%%%%%%%%%%%%%%%%%%%%%%%%%%%%%%%%%%%%%%%%%%%%%%%%%%%%

J.S.N., M.D.R., and F.J.B. formulated the research question.
J.S.N. wrote the code, conducted synthetic examples, and performed the main analysis.
J.S.N, M.D.R, and F.J.B wrote and reviewed the manuscript.

%%%%%%%%%%%%%%%%%%%%%%%%%%%%%%%%%%%%%%%%%%%%%%%%%%%%%%%%%%%%%%%%%%%%%%
\section{Acknowledgments}
%%%%%%%%%%%%%%%%%%%%%%%%%%%%%%%%%%%%%%%%%%%%%%%%%%%%%%%%%%%%%%%%%%%%%%
All code is written in Julia \citep{Bezanson2017} and is available publicly on GitHub at \url{https://github.com/jsnowynorth/BayesianSVD.jl}.
Additionally, the Bayesian SVD model has been developed into a Julia package \textit{BayesianSVD.jl}, which can also be downloaded from GitHub at \url{https://github.com/jsnowynorth/BayesianSVD.jl} 
for easy use.
All data are publicly available at the Copernicus Climate Data Store \url{https://cds.climate.copernicus.eu/cdsapp#!/dataset/reanalysis-era5-single-levels?tab=overview}.

This research was supported by the Director, Office of Science, Office of Biological and Environmental Research of the U.S.\ Department of Energy under Contract No.\ DE-AC02-05CH11231 and by the Regional and Global Model Analysis Program area within the Earth and Environmental Systems Modeling Program.
This research was also supported by the National Science Foundation [OAC-1931363, ACI-1553685] and the National Institute of Food \& Agriculture [COL0-FACT-2019].
The research used resources of the National Energy Research Scientific Computing Center (NERSC), also supported by the Office of Science of the U.S.\ Department of Energy, under Contract No.\ DE-AC02-05CH11231.

This document was prepared as an account of work sponsored by the United States Government. While this document is believed to contain correct information, neither the United States Government nor any agency thereof, nor the Regents of the University of California, nor any of their employees, makes any warranty, express or implied, or assumes any legal responsibility for the accuracy, completeness, or usefulness of any information, apparatus, product, or process disclosed, or represents that its use would not infringe privately owned rights. Reference herein to any specific commercial product, process, or service by its trade name, trademark, manufacturer, or otherwise, does not necessarily constitute or imply its endorsement, recommendation, or favoring by the United States Government or any agency thereof, or the Regents of the University of California. The views and opinions of authors expressed herein do not necessarily state or reflect those of the United States Government or any agency thereof or the Regents of the University of California.

%%%%%%%%%%%%%%%%%%%%%%%%%%%%%%%%%%%%%%%%%%%%%%%%%%%%%%%%%%%%%%%%%%%%%%
\setstretch{1}
\bibliographystyle{apalike}
\bibliography{referencesClean}
%%%%%%%%%%%%%%%%%%%%%%%%%%%%%%%%%%%%%%%%%%%%%%%%%%%%%%%%%%%%%%%%%%%%%%

\newpage

%%%%%%%%%%%%%%%%%%%%%%%%%%%%%%%%%%%%%%%%%%%%%%%%%%%%%%%%%%%%%%%%%%%%%%
\setcounter{section}{0}
\setcounter{equation}{0}
\renewcommand{\thesection}{A}
\renewcommand{\theequation}{A.\arabic{equation}}
\section{Appendix}\label{sec:appendix}
%%%%%%%%%%%%%%%%%%%%%%%%%%%%%%%%%%%%%%%%%%%%%%%%%%%%%%%%%%%%%%%%%%%%%%

%%%%%%%%%%%%%%%%%%%%%%%%%%%%%%%%%%%%%%%%%%%%%%%%%%%%%%%%%%%%%%%%%%%%%%
\section{Proofs of propositions}\label{sec:props}
%%%%%%%%%%%%%%%%%%%%%%%%%%%%%%%%%%%%%%%%%%%%%%%%%%%%%%%%%%%%%%%%%%%%%%

We now prove the propositions describing the properties of the orthogonal matrix constructed in section~\ref{sec:GeneratingMechanism}.

%%%%%%%%%%%%%%%%%%%%%%%%%%%%%%%%%%%%%%%%%%%%%%%%%%%%%%%%%%%%%%%%%%%%%%
\noindent \rule{\linewidth}{0.1em}

\begin{lemma}\label{lemma:one}
    The generating random variables $\vz_j$ and $\vOmega_j$ are exchangeable. 
\end{lemma}

\begin{proof}

The generating random variables $\vz_i$ are exchangeable because they all independent and have the same marginal distribution.
Specifically, because $\vOmega_1, \ldots, \vOmega_k \sim \pi_{\Omega}$ all have the same distribution, if we marginalize $\vz_j$, we get $p(\vz_j) = \int_{\Omega} p(\vz_j|\vOmega_j)p(\vOmega_j) d\vOmega$ is the same for all $j = 1, \ldots, k$.
    
\end{proof}

%%%%%%%%%%%%%%%%%%%%%%%%%%%%%%%%%%%%%%%%%%%%%%%%%%%%%%%%%%%%%%%%%%%%%%
\noindent \rule{\linewidth}{0.1em}

\begin{lemma}\label{lemma:two}
    For any permutation $\pi$ of the columns of the $n \times k$ matrix $\vX$, denoted $\vX_{\pi}$, the matrix $\vP_{\pi} \equiv \vI - \vX_{\pi}(\vX_{\pi}' \vX_{\pi})^{-1} \vX_{\pi}'$ is the unique projection onto the orthogonal complement of column space of $\vX$. That is, $\vP_{\pi} = \vP$.
\end{lemma}

\begin{proof}

Since $\vX$ and $\vX_{\pi}$ share the same column space, the result is immediate by the projection theorem.
\end{proof}

%%%%%%%%%%%%%%%%%%%%%%%%%%%%%%%%%%%%%%%%%%%%%%%%%%%%%%%%%%%%%%%%%%%%%%
\noindent \rule{\linewidth}{0.1em}
\noindent \textbf{Proposition 1.} 
\textit{
    The columns of $\vW = \vX (\vX' \vX)^{-1/2}$ are exchangeable. That is, for any permutation $\pi$ of the set $\{1, \ldots, k\}$, $p([\vw_1, \ldots, \vw_k]) \overset{d}{=} p([\vw_{\pi(1)}, \ldots, \vw_{\pi(k)}])$.
}

% \url{https://www2.stat.duke.edu/~sayan/CBB2012/exchange.pdf} see page 3/4.

\begin{proof}

We first show the columns of the matrix $\vX$ are exchangeable. 
That is, for any permutation $\pi$ of the set $\{1, \ldots, k\}$, $p([\vx_1, \ldots, \vx_k]) \overset{d}{=} p([\vx_{\pi(1)}, \ldots, \vx_{\pi(k)}])$.
Then, we use the exchangeability of $\vX$ to show exchangeability of $\vW$.

Define $\vX_{\pi_j} = [\vx_{\pi(1)}, \ldots, \vx_{\pi(j)}]$ and $\vP_{\pi(j)} = \vI - \vX_{\pi_j} (\vX_{\pi_j}'\vX_{\pi_j})^{-1} \vX_{\pi_j}' = \vP_j$.
To show exchangeablility, we show the characteristic function of $\vX$ is equivalent to the characteristic function of $\vX_{\pi_j}$.
For a $n \times k$ random matrix $\vX$, the characteristic function is defined as $\varphi(\vX) = E[\exp \{ i \mbox{tr} (\vT' \vX) \}] = E[\exp \{ i \sum_{\ell = 1}^k \vt_\ell' \vx_\ell \}]$, where $\vT = [\vt_1, \ldots, \vt_k]$ is a $n \times k$ matrix, $i$ is the imaginary unit, and tr$(\cdot)$ is the trace operator.
We show the proposition using proof by induction:
\begin{enumerate}
    \item For $k=1$, we have $\vX_{\pi_1} = \vx_{\pi(1)} = \vP_0\vz_{\pi(1)} \overset{d}{=} \vP_0\vz_1 = \vx_1 = \vX_1$, where $\vz_{\pi(1)} \overset{d}{=} \vz_1$ by lemma \ref{lemma:one}. Therefore, $\vX_1 \overset{d}{=} \vX_{\pi_1}$.

    \item Assume for $k=j$, $\vX_j \overset{d}{=} \vX_{\pi_j}$.

    \item By the characteristic function of $\vX_{\pi_{j+1}}$,
    \begin{align*}
        \varphi(\vX_{\pi_{j+1}}) & = E\left[ \exp\left\{ i \sum_{\ell=1}^{j+1}\vt_{\ell}'\vx_{\pi(\ell)} \right\} \right] \\
            & = E\left[ \exp\left\{ i \sum_{\ell=1}^{j}\vt_{\ell}'\vx_{\pi(\ell)} \right\} \exp\left\{ i \vt_{j+1}'\vx_{\pi(j+1)}  \right\} \right] \\
            & = E\left[ \exp\left\{ i \sum_{\ell=1}^{j}\vt_{\ell}'\vx_{\pi(\ell)} \right\} E[\exp\left\{ i \vt_{j+1}'\vx_{\pi(j+1)}  \right\} |\vX_{\pi_j}, \vOmega_{\pi(j+1)}]\right] \quad \mbox{\footnotesize{(iterative expectation)}} \\
            & = E\left[ \exp\left\{ i \sum_{\ell=1}^{j}\vt_{\ell}'\vx_{\pi(\ell)} \right\} E[\exp\left\{ i \vt_{j+1}'\vP_{\pi(j)}\vz_{\pi(j+1)}  \right\} |\vX_{\pi_j}, \vOmega_{\pi(j+1)}]\right] \\
            & = E\left[ \exp\left\{ i \sum_{\ell=1}^{j}\vt_{\ell}'\vx_{\pi(\ell)} \right\} \exp\left\{ \vt_{j+1}'\vP_{\pi(j)} \vOmega_{\pi(j+1)} \vP'_{\pi(j)} \vt_{j+1} \right\}\right].
    \end{align*}

    The induction hypothesis implies $\{\vX_j, \vP_j\} \overset{d}{=} \{\vX_{\pi_j}, \vP_{\pi(j)}\}$.
    Also, $\vOmega_{\pi(j+1)}$ is independent of $\vX_j$ and $\vP_j$. 
    Therefore, $\{\vX_j, \vP_j, \vOmega_{j+1}\} \overset{d}{=} \{\vX_{\pi_j}, \vP_{\pi(j)}, \vOmega_{\pi(j+1)}\}$ because $\vX_{\pi_j} \Rightarrow \vX_j$ by induction hypothesis, $\vP_{\pi(j)} \equiv \vP_j$ by lemma \ref{lemma:two}, and $\vOmega_{\pi(j+1)} \Rightarrow \vOmega_{j+1}$ because it is independent of $\vX_j$ and $\vP_j$ and it is exchangeable.
    Thus,    
    \begin{align*}
        \varphi(\vX_{\pi_{j+1}}) & = E\left[ \exp\left\{ i \sum_{\ell=1}^{j}\vt_{\ell}'\vx_{\ell} \right\} \exp\left\{ \vt_{j+1}'\vP_{j} \vOmega_{j+1} \vP'_{j} \vt_{j+1} \right\}\right] = \varphi(\vX_{j+1}),
    \end{align*}
    and $\vX_{j+1} \overset{d}{=} \vX_{\pi_{j+1}}$.

    The exchangeability of $\vW$ follows from the exchangeability of $\vX$.
    Specifically, because the diagonal matrix $\vR \equiv (\vX' \vX)^{-1/2} = \text{diag}[(\vx'_1 \vx_1)^{-1/2}, \ldots, (\vx'_k \vx_k)^{-1/2}] \equiv \text{diag}[r_1, \ldots, r_k]$ where the elements $r_1, \ldots, r_k$ are the norm of the random vectors $\vx_1, \ldots, \vx_k$, respectively, is simply a rescaling of the columns of $\vX$, and the permutation of the scaling is preserved, $\vW_{\pi} \overset{d}{=} \vW$.
\end{enumerate}

\end{proof}

\noindent \rule{\linewidth}{0.1em}

%%%%%%%%%%%%%%%%%%%%%%%%%%%%%%%%%%%%%%%%%%%%%%%%%%%%%%%%%%%%%%%%%%%%%%
\noindent \textbf{Proposition 2.} 
\textit{ 
    $\vw_i|\vW_{i-1} \overset{d}{=} \vN_{i-1}\widetilde{\vw}_i|\vW_{i-1}$ where the columns of $\vN_{i-1}$ form an orthonormal basis for the null space of $\vW_{i-1}$ and $\widetilde{\vw}_i|\vW_{i-1} \sim \mbox{PN}_{n-i+1}(\vec{0}, \vN_{i-1}'\vOmega_i\vN_{i-1})$ is the projected weight function.
}

\begin{proof}
The following argument is similar to \citet{Hoff2007}, except now we account for dependence structure and the resulting distribution is different.
By construction, $\vw_i = \vP_{i-1}\vz_i/(\vz_i'\vP_{i-1}'\vP_{i-1}\vz_i)^{1/2}$ where $\vP_{i-1}$ has $n-i+1$ eigenvalues equal to 1 and the rest being 0.
Let the eigenvalue decomposition be $\vP_{i-1} = \vN_{i-1}\vN_{i-1}'$ where $\vN_{i-1}$ is an $n \times (n-i+1)$ matrix whose columns span the null space of $\vW_i$.
Making the substitution $\vP_{i-1} = \vN_{i-1}\vN_{i-1}'$,
\begin{align*}
    \vw_i & = \frac{\vP_{i-1}\vz_i}{(\vz_i'\vP_{i-1}'\vP_{i-1}\vz_i)^{1/2}} \\
          & = \frac{\vN_{i-1}\vN_{i-1}'\vz_i}{(\vz_i'\vN_{i-1}'\vN_{i-1}\vN_{i-1}\vN_{i-1}'\vz_i)^{1/2}} \\
          & = \vN_{i-1}\frac{\vN_{i-1}'\vz_i}{(\vz_i'\vN_{i-1}\vN_{i-1}'\vz_i)^{1/2}}.
\end{align*}
Note that $\vP_{i-1} = \vI - \vW_{i-1}\vW_{i-1}'$, so $\vw_i|\vW_{i-1} \overset{d}{=} \vN_{i-1}\frac{\vN_{i-1}'\vz_i}{(\vz_i'\vN_{i-1}\vN_{i-1}'\vz_i)^{1/2}}$.
Because $\vz_i \sim N_n(\vec{0}, \vOmega_i)$, we have $\vN_{i-1}'\vz_i|\vW_{i-1} \sim N_n(\vec{0}, \vN_{i-1}'\vOmega_i\vN_{i-1})$ and $\frac{\vN_{i-1}'\vz_i}{(\vz_i'\vN_{i-1}\vN_{i-1}'\vz_i)^{1/2}} \equiv \widetilde{\vw}_i|\vW_{i-1} \sim PN(\vec{0}, \vN_{i-1}'\vOmega_i\vN_{i-1})$.
\end{proof}

\newpage

%%%%%%%%%%%%%%%%%%%%%%%%%%%%%%%%%%%%%%%%%%%%%%%%%%%%%%%%%%%%%%%%%%%%%%
\setcounter{section}{0}
\renewcommand{\thesection}{S}
\renewcommand{\thetable}{S}
\renewcommand{\thefigure}{S.\arabic{figure}}
\renewcommand{\theequation}{S.\arabic{equation}}
\setcounter{figure}{0}
\section{Supplemental Material}\label{sec:supplement}
%%%%%%%%%%%%%%%%%%%%%%%%%%%%%%%%%%%%%%%%%%%%%%%%%%%%%%%%%%%%%%%%%%%%%%

%%%%%%%%%%%%%%%%%%%%%%%%%%%%%%%%%%%%%%%%%%%%%%%%%%%%%%%%%%%%%%%%%%%%%%
\subsection{Full conditional distributions}\label{sec:FCD}
%%%%%%%%%%%%%%%%%%%%%%%%%%%%%%%%%%%%%%%%%%%%%%%%%%%%%%%%%%%%%%%%%%%%%%

The diagonal matrix $\vD$ and the length-scale parameters $\rho_{u,i}$ and $\rho_{v,i}$ do not have conjugate updates and so we use a Metropolis-within-Gibbs step to estimate these parameters.
For all Metropolis steps, we use a truncated normal for the proposal distribution with the mean set to the most recently accepted value.
For $\vD$ the upper truncation bound is set to infinity and for $\rho_{u,i}$ and $\rho_{v,i}$ the upper truncation bound is set to the max distance for $\vU$ (e.g., greatest distance between spatial locations) and $\vV$ (e.g., greatest span between time points) divided by 2, respectively.
Because the variance of the proposal can influence the acceptance rate, we automatically tune the proposal variance for each parameter individually such that the acceptance rate is between 25\% and 45\%.

\vspace{2em}

\textbf{Sampling Algorithm}

For each iteration of the MCMC algorithm, do:
\begin{enumerate}
    \item Update $\vD$ using a Metropolis step
    
    \item For $i \in \{1, \ldots, k\}$ update $\vu_i|\vU_{-i} \overset{d}{=} \vN^u_{i} \widetilde{\vu}_i$ where $\widetilde{\vu}_i$
        \begin{align*}
            [\widetilde{\vu}_i|\cdot]  & \sim N(\vS_u^{-1}\vm_u, \vS_u^{-1}) \mathbb{I}(\widetilde{\vu}_i \in \mathcal{V}_{1,n}) \\
                              & \vm_u = d_i \vN_i^{u \trp} \vSigma^{-1} \vE_i \vv_i \nonumber \\
                              & \vS_u = d_i^2 (\vN^{u \trp}_{i}\vOmega^u_i\vN^u_{i})^{-1} + d_i^2 \vN_i^{u \trp} \vSigma^{-1} \vN_i^{u}. \nonumber
        \end{align*}
    \item For $i \in \{1, \ldots, k\}$ update $\vv_i|\vV_{-i} \overset{d}{=} \vN^v_{i} \widetilde{\vv}_i$ where $\widetilde{\vv}_i$
        \begin{align*}
            [\widetilde{\vv}_i|\cdot] & \sim N(\vS_v^{-1}\vm_v, \vS_v^{-1}) \mathbb{I}(\widetilde{\vv}_i \in \mathcal{V}_{1,m}) \\
                              & \vm_v = d_i \vN_i^{v \trp} \vE_i' \vSigma^{-1} \vu_i \nonumber\\
                              & \vS_v = d_i^2 (\vN^{v \trp}_{i}\vOmega^v_i\vN^v_{i})^{-1} + d_i^2 \vu_i' \vSigma^{-1} \vu_i \vI_m. \nonumber
        \end{align*}
    
    \item Recall, we parameterize $\vSigma = \sigma^2 \vI_n$.
            The full conditional distribution for $\sigma^2$ is
            \begin{align*}
                [a|\cdot] & \sim IG((\xi + 1)/2, (1/A^2) + \xi/\sigma) \\
                [\sigma^2|\cdot] & \sim IG((nm + \xi)/2, \xi/a + \mbox{vec}(\vZ - \vU \vD \vV')'\mbox{vec}(\vZ - \vU \vD \vV')/2).
            \end{align*}
            We specify $\xi = 1$ and $A = 10^5$ which corresponds to the prior $\sigma \sim \mbox{Half-}t(\xi, A) \equiv \mbox{Half-cauchy}(A)$.
    
    \item Recall, we parameterize $\vOmega_u(\theta_{u,i}) = \sigma^2_{u,i} \vC_u(\theta_{u,i})$.
            For $i \in \{1, \ldots, k\}$ update $\sigma^2_{u,i}$ from
            \begin{align*}
                [a_{u,i}|\cdot] & \sim IG((\xi + 1)/2, (1/A^2) + \xi/\sigma_{u,i}) \\
                [\sigma^2_{u,i}|\cdot] & \sim IG((n - k + 1 + \xi)/2, \xi/a_{u,i} + (d_i^2 \widetilde{\vu}_i' (\vN^{u \trp}_{i}\vOmega^u_i\vN^u_{i})^{-1} \widetilde{\vu}_i)/2),
            \end{align*}
            with $\xi = 1$ and $A = 10^5$.
    
    \item Recall, we parameterize $\vOmega_v(\theta_{v,i}) = \sigma^2_{v,i} \vC_v(\theta_{v,i})$.
            For $i \in \{1, \ldots, k\}$ update $\sigma^2_{v,i}$ from
            \begin{align*}
                [a_{v,i}|\cdot] & \sim IG((\xi + 1)/2, (1/A^2) + \xi/\sigma_{v,i}) \\
                [\sigma^2_{v,i}|\cdot] & \sim IG((m - k + 1 + \xi)/2, \xi/a_{v,i} + (d_i^2 \widetilde{\vv}_i' (\vN^{v \trp}_{i}\vOmega^v_i\vN^v_{i})^{-1} \widetilde{\vv}_i)/2),
            \end{align*}
        with $\xi = 1$ and $A = 10^5$.
    
    \item For $i \in \{1, \ldots, k\}$ update $\rho^2_{u,i}$ using a Metropolis step
    
    \item For $i \in \{1, \ldots, k\}$ update $\rho^2_{v,i}$ using a Metropolis step
\end{enumerate}

%%%%%%%%%%%%%%%%%%%%%%%%%%%%%%%%%%%%%%%%%%%%%%%%%%%%%%%%%%%%%%%%%%%%%%
\subsection{Identity correlation}\label{sec:IC}
%%%%%%%%%%%%%%%%%%%%%%%%%%%%%%%%%%%%%%%%%%%%%%%%%%%%%%%%%%%%%%%%%%%%%%

When $\vOmega_i = \vI_n$, $\widetilde{\vw}_k$ in proposition 2 is uniformly distributed on the Stiefel manifold.
To see this, note that for $\vz_k \sim N(0, \vI)$, $\vN_{k-1}'\vz_k \sim \mbox{N}_{n-k+1}(0, \vI)$, and $\frac{\vN_{k-1}'\vz_k}{(\vz_i'\vN_{i-1}\vN_{i-1}'\vz_i)^{1/2}}$ is uniformly distributed on the $n-k+1$ sphere.
Also, we see proposition 1 is now equivalent to \citet{Hoff2007}, and $\vW$ is the uniform probability measure on $\mathcal{V}_{k,n}$.

The resulting full conditional distributions for $\widetilde{\vu}_i$ and $\widetilde{\vv}_i$ when $\vOmega^u_i = \vI_{n}$ and $\vOmega^v_i = \vI_{m}$ for the SVD model in~\ref{sec:FCD} become the von-Mises Fisher distribution, which is equivalent to the full conditionals of \citet{Hoff2007}.
To see this, note the mean of $[\widetilde{\vu}_i|\cdot]$ is $\vS_u^{-1}\vm_u = \frac{1}{\sigma^2 + 1}\frac{1}{d_i} \vN_i^{u \trp} \vE_i \vv_i$ and the covariance is $\vS_u^{-1} = \left(d_i^2 + \frac{d_i^2}{\sigma^2}\right)^{-1}$.
The full conditional distribution
\begin{align*}
    [\widetilde{\vu}_i|\cdot] & \propto \exp\left\{-\frac{1}{2}tr\left[-2\widetilde{\vu}'_i d_i \vN_i^{u \trp} \vE_i \vv_i \left( \frac{1}{\sigma^2 + 1} + \frac{1}{\sigma^4 + \sigma^2} \right)\right] \right\} \\
    & = \exp\left\{-\frac{1}{2}tr\left[-2\widetilde{\vu}'_i d_i \vN_i^{u \trp} \vE_i \vv_i \left( \frac{1}{\sigma^2} \right)\right] \right\},
\end{align*}
which is the kernel of the von-Mises Fisher distribution.
The same result holds for $\widetilde{\vv}_i$.

%%%%%%%%%%%%%%%%%%%%%%%%%%%%%%%%%%%%%%%%%%%%%%%%%%%%%%%%%%%%%%%%%%%%%%
\subsubsection{Relationship to algorithmic SVD}\label{sec:ASVD}

Computing the SVD of $\vY$,
\begin{align*}
    \vY & = \vU \vD \vV' \\
    \vY & = \vU_{-i} \vD_{-i} \vV'_{-i} + d_{i} \vu_{i} \vv'_{i} \\
    \vY - \vU_{-i} \vD_{-i} \vV'_{-i} & = \vE_{-i} =  d_{i} \vu_{i} \vv'_{i} \\
    \frac{1}{d_i}\vE_{-i}\vv_{i} & = \vu_{i},
\end{align*}
so the $i$th basis function can be expressed as a function of the data and other basis functions.
The mean of the full conditional distribution $[\widetilde{\vu}_i|\cdot]$ is $\vS_u^{-1}\vm_u = \frac{1}{\sigma^2 + 1}\frac{1}{d_i} \vN_i^{u \trp} \vE_i \vv_i$, and $E[\widetilde{\vu}_i|\cdot] \rightarrow \frac{1}{d_i} \vN_i^{u \trp} \vE_i \vv_i$ as $\sigma^2 \rightarrow 0$.
Mapping to the original space, $\vN_i^{u} E[\widetilde{\vu}_i|\cdot] = \frac{1}{d_i} \vN_i^{u} \vN_i^{u \trp} \vE_i \vv_i = \frac{1}{d_i} \vE_i \vv_i$.
While not shown here, the same argument applies for $\vV$.
Therefore, when the covariance is taken to be the identity matrix, the posterior mean of the basis functions is equivalent to the C-SVD basis functions.

To see the relationship, we repeat one of the simulation conducted in Section~\ref{sec:rankSimulation} with $SNR = 5$, $k = 5$, and set the correlation matrices $\vC_u$ and $\vC_v$ to be the identity.
Here, we still estimate the basis function specific variance $\sigma^2_{u,i}$ and $\sigma^2_{v,i}$.
We obtain 10000 samples from the posterior, discarding the first 5000 as burnin.
The resulting estimates for the $\vU$ and $\vV$ basis functions are shown in Figure~\ref{fig:IdentityBasis}, where the posterior mean of the basis functions (blue) is nearly identical to the C-SVD estimates (red).
In all cases, we see the 95\% intervals (blue shaded region) cover the C-SVD estimates but has $\approx$95\% coverage of the true line (black).

\begin{figure}
    \centering
    \includegraphics[width = \linewidth]{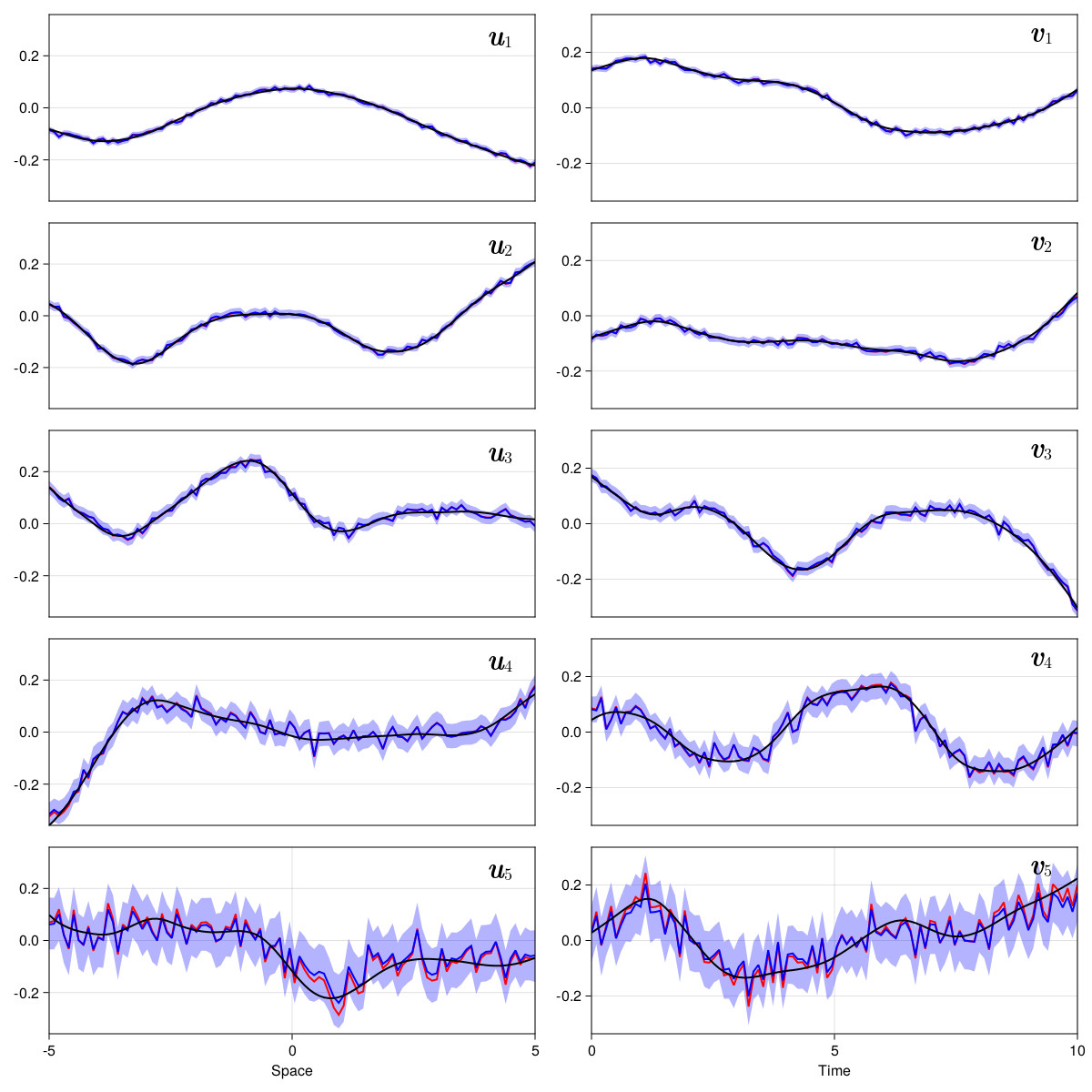}
    \caption{Posterior mean (blue line), 95\% credible intervals (shaded blue region), truth (black line), and C-SVD estimate (red line) for the $\vU$ and $\vV$ basis function.}
    \label{fig:IdentityBasis}
\end{figure}

\FloatBarrier

\subsection{Computation and scalability}\label{sec:scale}

While the proposed prior is relatively simple to specify and implement, there are some computational aspects to consider.
On one hand, the fact that the prior is conjugate with a normal data distribution means that MCMC updates for the columns of $\vU$ and $\vV$ can be obtained in a straightforward manner.
On the other hand, calculating the full conditional distributions (from which the Gibbs draws are sampled) is computationally intensive for large $n$.  
From the formulation in Section~\ref{sec:PSVD}, the full conditional distributions for the columns of $\vU$ and $\vV$ involve matrix inverses $(\vN^{u \trp}_i \vOmega^u_i \vN^u_i)^{-1}$ and $(\vN^{v \trp}_i \vOmega^v_i \vN^v_i)^{-1}$, respectively (see supplement \ref{sec:FCD}), each of which are dense $(n-k+1)\times(n-k+1)$ and $(m-k+1)\times(m-k+1)$ matrices, respectively.
Therefore, in order to update $\vU$ and $\vV$ once in an MCMC iteration, we need to calculate $2k$ matrix inverses (one for each of the $k$ columns of $\vU$ and $\vV$), which is computationally challenging for large $n$ or $m$.
Furthermore, updating the hyperparameters of the kernel (e.g., the length-scale parameters $\rho_{u,i}$ and $\rho_{v,i}$) requires Metropolis-Hastings steps.
In this case, the likelihood involves a multivariate Normal density: when the covariance of the multivariate Normal is non-diagonal and dense (as is the case here), the number of flops associated with evaluating the determinant and solving quadratic forms scales with $\mathcal{O}(n^3)$.
Again, each iteration of the MCMC requires $2k$ of these calculations.
As such, without significant computing resources, the required computation for the model as-is proves difficult for data where either $n$ or $m$ is greater than 1000. 
More specifically, Figure~\ref{fig:scalability} shows an estimate of the amount of time needed to update all parameters in a single iteration of the MCMC for the special case of $k=5$ across different sample sizes $n$ and $m$ on a personal laptop.

\begin{figure}[t]
    \centering
    \includegraphics[width = \linewidth]{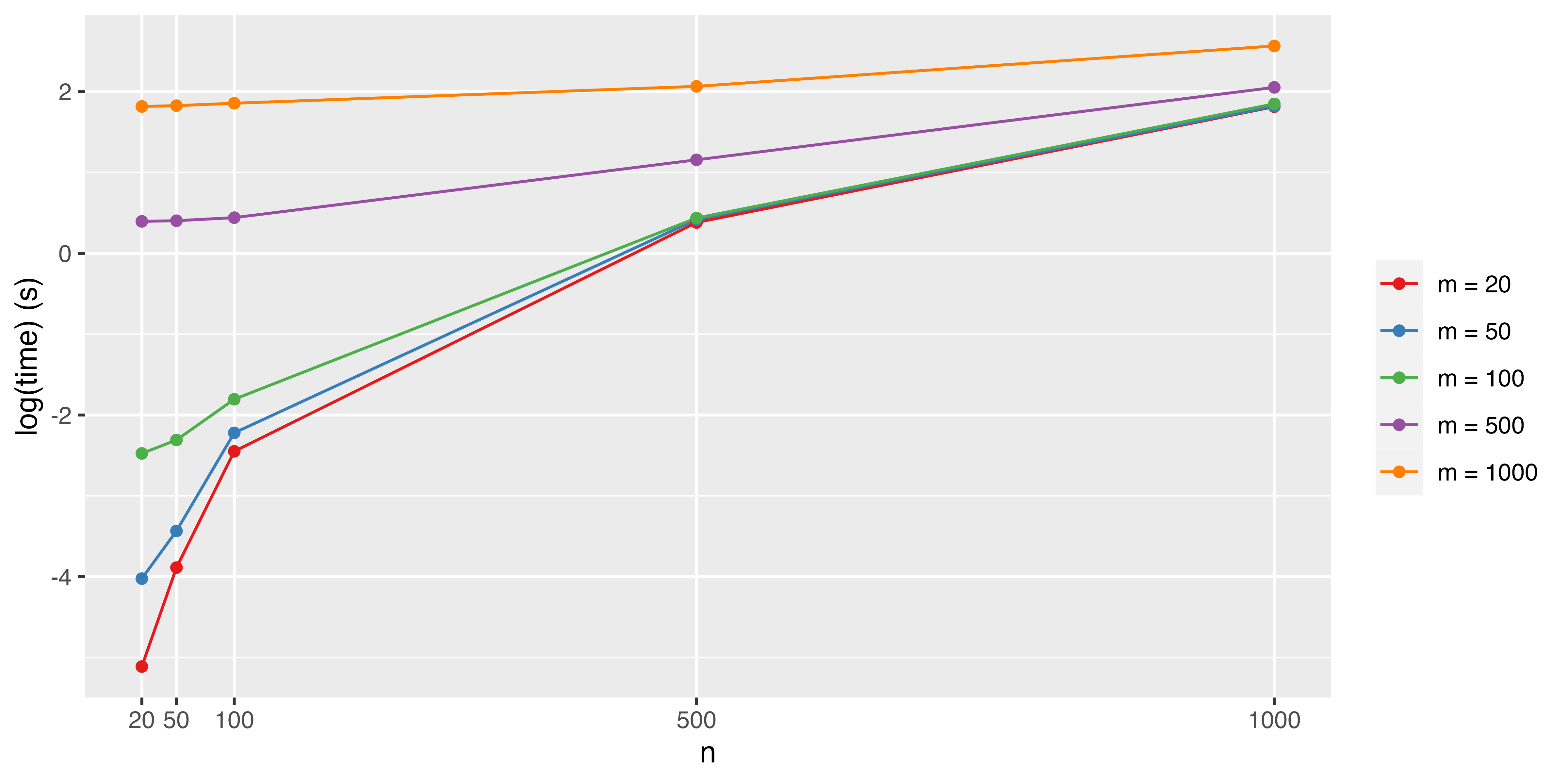}
    \caption{Median log computation time for a single MCMC iteration as a function of sample size $n$ (x-axis) and $m$ (line color).}
    \label{fig:scalability}
\end{figure}

In spite of these apparent limitations, there are a variety of approaches one could take to reduce the associated computational burdens of this model.
The simplest approach would be to parameterize the covariance matrix as $\vOmega_i = \sigma^2_i\vC_i(\vtheta)$ where $\vtheta$ is specified and not estimated: this would remove the Metropolis-Hastings steps required to estimate $\rho_{u,i}$ and $\rho_{v,i}$.
Alternatively, one could use a compactly-supported kernel function \citep[see, e.g.,][]{Wendland1995, Buhmann2000, Genton2000, Gneiting2002, Melkumyan2009} and leverage sparse matrix techniques. 
These approaches are targeted at reducing the cost of estimating $\vC_i(\vtheta_i)$ and the associated parameters.
In either case, however, our proposed prior does have the disadvantage of relying on a column-wise sampling strategy and the corresponding matrix calculations needed to sample each column.
Specifically, within each MCMC iteration, there is a required $\mathcal{O}(n^3)$ cost of computing the orthonormal basis for the null-space $\vN_i^u$ and $\vN_i^v$ and the ensuing inverses.
In other words, we must calculate the inverse of $\vN^{(\cdot) \trp}_i \vOmega^{(\cdot)}_i \vN^{(\cdot)}_i$ which is dense irrespective of the sparsity of $\vOmega^{(\cdot)}_i$.
For this reason, implementing sparse matrix techniques for the $\vOmega_i$ will not solve this challenge.

\FloatBarrier

%%%%%%%%%%%%%%%%%%%%%%%%%%%%%%%%%%%%%%%%%%%%%%%%%%%%%%%%%%%%%%%%%%%%%%
\subsection{Projected normal distribution}\label{sec:PND}
%%%%%%%%%%%%%%%%%%%%%%%%%%%%%%%%%%%%%%%%%%%%%%%%%%%%%%%%%%%%%%%%%%%%%%

Let $\vz_j \sim N_n(\vec{0}, \vOmega), j = 1, \ldots, K$ and $\vZ = [\vz_1, \ldots, \vz_K]$.
Then $\vW = \vZ/ \|\vZ\| \sim PN(\vec{0}, \vOmega)$ where $\vw_j$ is a length-$n$ directional vector with $n-1$ angles $\vtheta_j = [\theta_{1,j}, \theta_{2,j}, \ldots, \theta_{n-1,j}]$.
Using spherical coordinates, $r_j = \|\vz_j\| = \sqrt{z_{1,j}^2 + \cdots + z_{n,j}^2}$,
\begin{align*}
    w_{1,j} & = \cos(\theta_{1, j}) \\
    w_{2,j} & = \sin(\theta_{1, j}) \cos(\theta_{2, j}) \\
     & \vdots \\
    w_{n-1,j} & = \sin(\theta_{1, j}) \cdots \sin(\theta_{n-2, j})\cos(\theta_{n-1, j}) \\
    w_{n,j} & = \sin(\theta_{1, j}) \cdots \sin(\theta_{n-2, j})\sin(\theta_{n-1, j})
\end{align*}
and
\begin{align*}
    z_{1,j} & = r_j \cos(\theta_{1, j}) \\
    z_{2,j} & = r_j \sin(\theta_{1, j}) \cos(\theta_{2, j}) \\
     & \vdots \\
    z_{n-1,j} & = r_j \sin(\theta_{1, j}) \cdots \sin(\theta_{n-2, j})\cos(\theta_{n-1, j}) \\
    z_{n,j} & = r_j \sin(\theta_{1, j}) \cdots \sin(\theta_{n-2, j})\sin(\theta_{n-1, j})
\end{align*}
with $r_j \geq 0$, $\theta_{1,j}, \theta_{2,j}, \ldots, \theta_{n-2,j} \in [0, \pi]$, and $\theta_{n-1,j} \in [0, 2\pi]$.
Augmenting the distribution with its latent length $r_j$, we get the joint density of $(r_j, \vw_j)$ is
\begin{align*}
    p(r_j, \vw_j) = (2 \pi)^{-n/2}|\vOmega|^{-1/2}\exp\left\{-\frac{1}{2}(r_j \vw_j)'\vOmega^{-1}(r_j \vw_j) \right\}r_i^{n-1} \mathbb{I}(\vw_i \in \mathcal{V}_{1,n}),
\end{align*}
where the area element on the unit sphere is $r_j^{n-1}sin^{n-2}(\theta_{1,j})sin^{n-3}(\theta_{2,j})\ldots sin(\theta_{n-2,j})dr_jd\theta_{1,j}d\theta_{2,j}\ldots d\theta_{n-1,j}$.
For properties of this distribution, see \citet{Hernandez-Stumpfhauser2017}.

%%%%%%%%%%%%%%%%%%%%%%%%%%%%%%%%%%%%%%%%%%%%%%%%%%%%%%%%%%%%%%%%%%%%%%
\subsection{Additional simulation figures}\label{sec:SimFigs}
%%%%%%%%%%%%%%%%%%%%%%%%%%%%%%%%%%%%%%%%%%%%%%%%%%%%%%%%%%%%%%%%%%%%%%

\begin{figure}[ht]
    \centering
    \includegraphics[width = 0.5\linewidth]{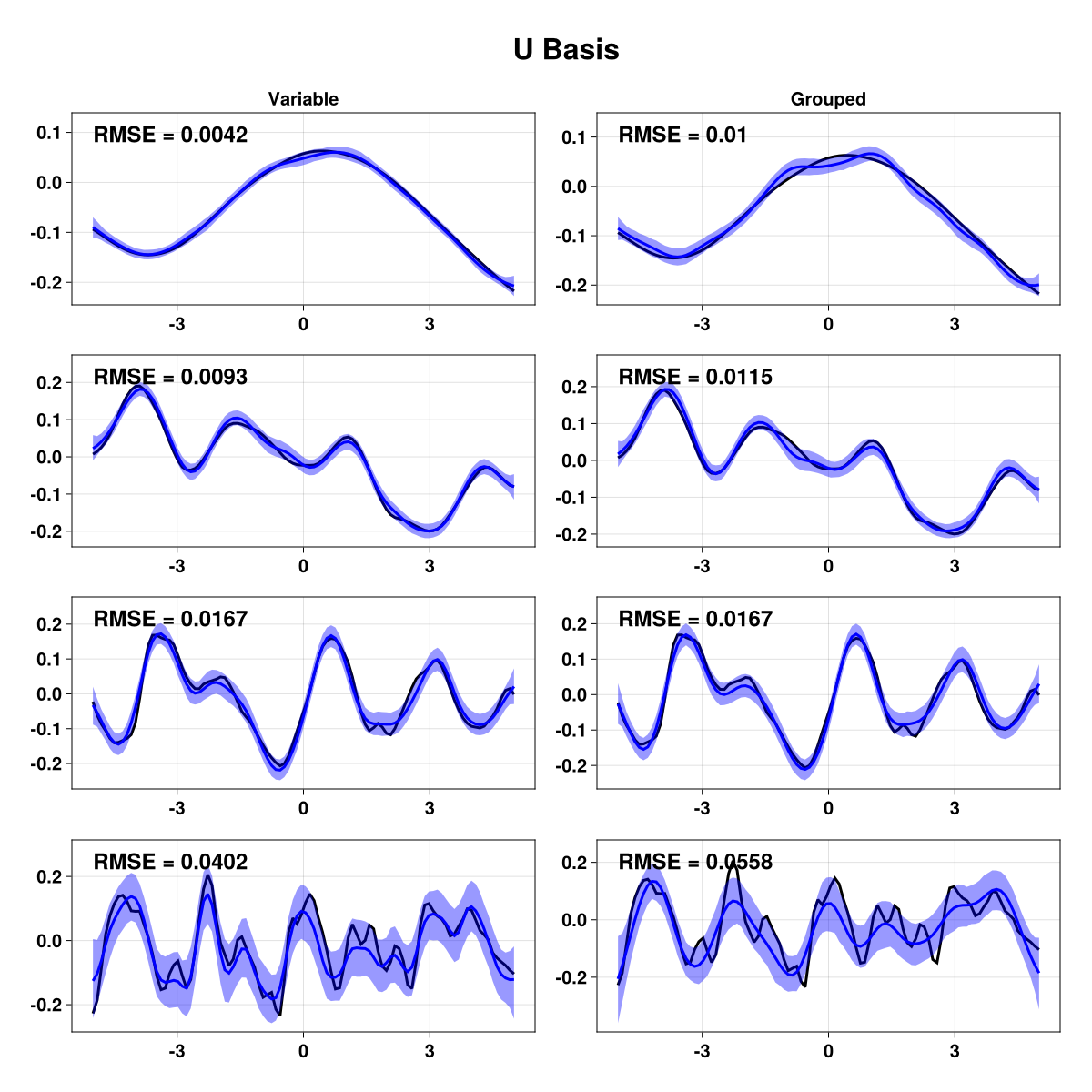}%
    \includegraphics[width = 0.5\linewidth]{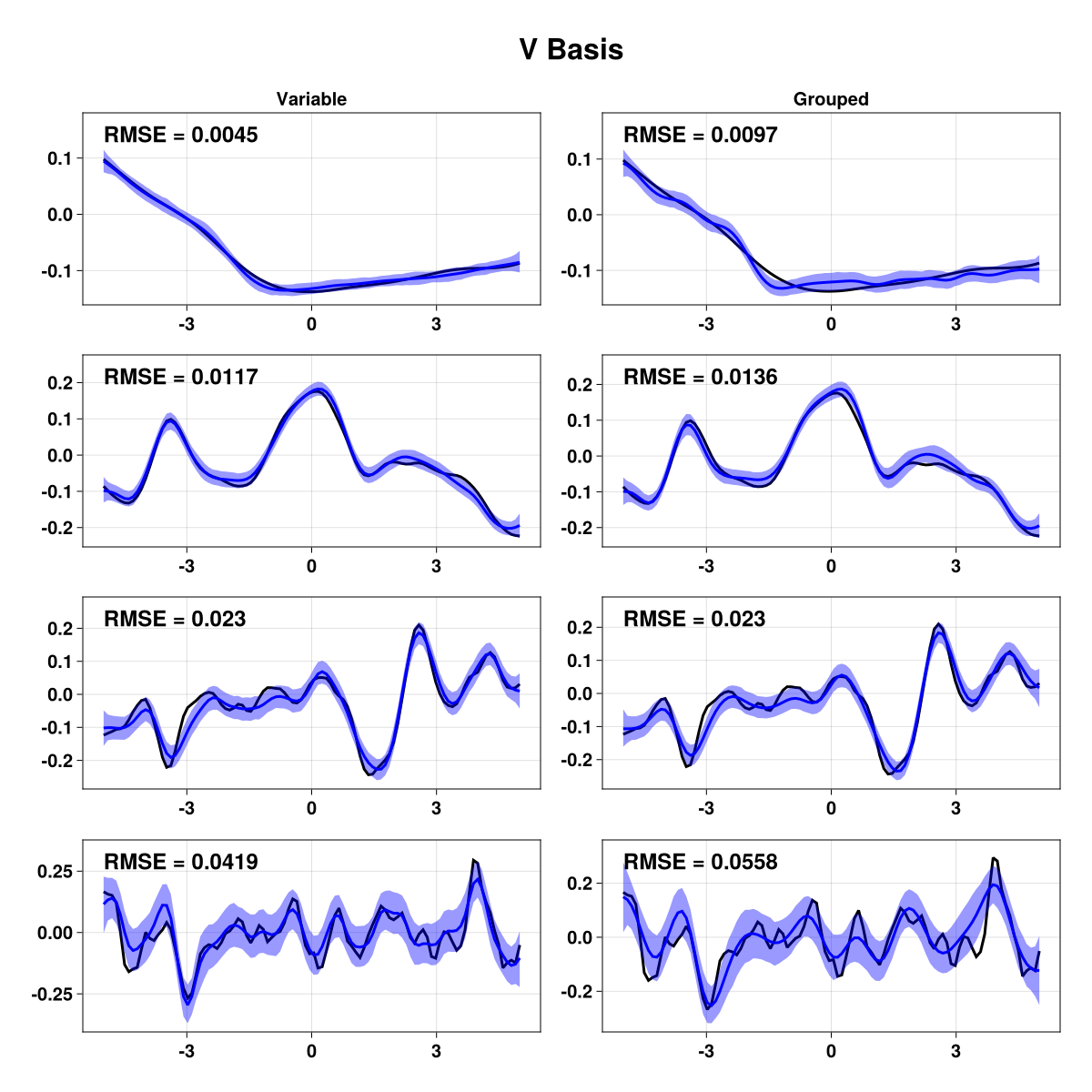}
    \caption{Randomly chosen results from one simulation described in Section \ref{sec:variablelength} with a SNR of 1. The $\vU$ ($\vV$) basis functions are shown on the left (right) and within the sub-plot the results from the variable (grouped) model are shown on the left (right) where the top row corresponds to the first basis function (e.g., $\vu_1$ or $\vv_1$) and the bottom row corresponds to the fourth basis function (e.g., $\vu_4$ or $\vv_4$). In each panel, the black line is the true basis function, blue line is the posterior mean, and blue shaded region denotes the 95\% credible interval (CI).}
    \label{fig:variableLengthPlot}
\end{figure}

\begin{figure}[ht]
    \centering
    \includegraphics[width = \textwidth]{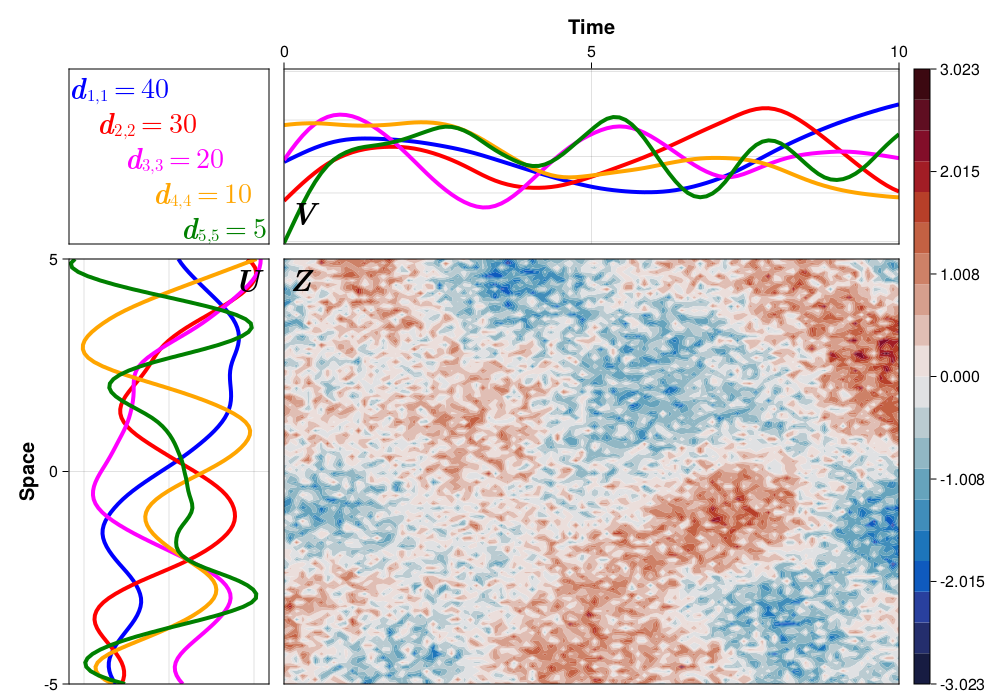}
    \caption{Randomly simulated data $\vZ$ (main plot) with a signal-to-noise ratio of 1 and the randomly simulated $\vU$ (left) and $\vV$ (top) basis functions.}
    \label{fig:data1D}
\end{figure}

\begin{figure}[t]
    \centering
    \includegraphics[width = \textwidth]{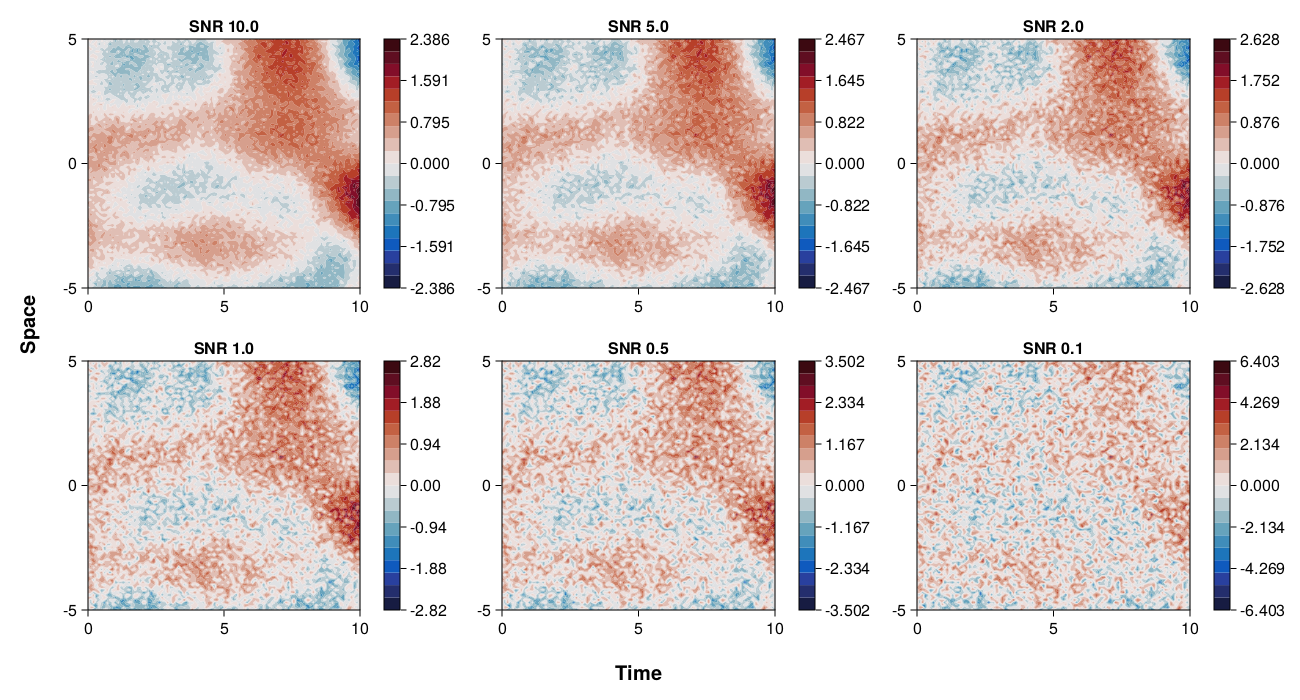}
    \caption{Example of simulated data with varying levels of signal-to-noise ratio.}
    \label{fig:simdata1D}
\end{figure}

\FloatBarrier

%%%%%%%%%%%%%%%%%%%%%%%%%%%%%%%%%%%%%%%%%%%%%%%%%%%%%%%%%%%%%%%%%%%%%%
\subsection{Surface air temperature}\label{sec:airTempSupp}
%%%%%%%%%%%%%%%%%%%%%%%%%%%%%%%%%%%%%%%%%%%%%%%%%%%%%%%%%%%%%%%%%%%%%%

\begin{figure}
    \centering
    \includegraphics[width = 0.9\linewidth]{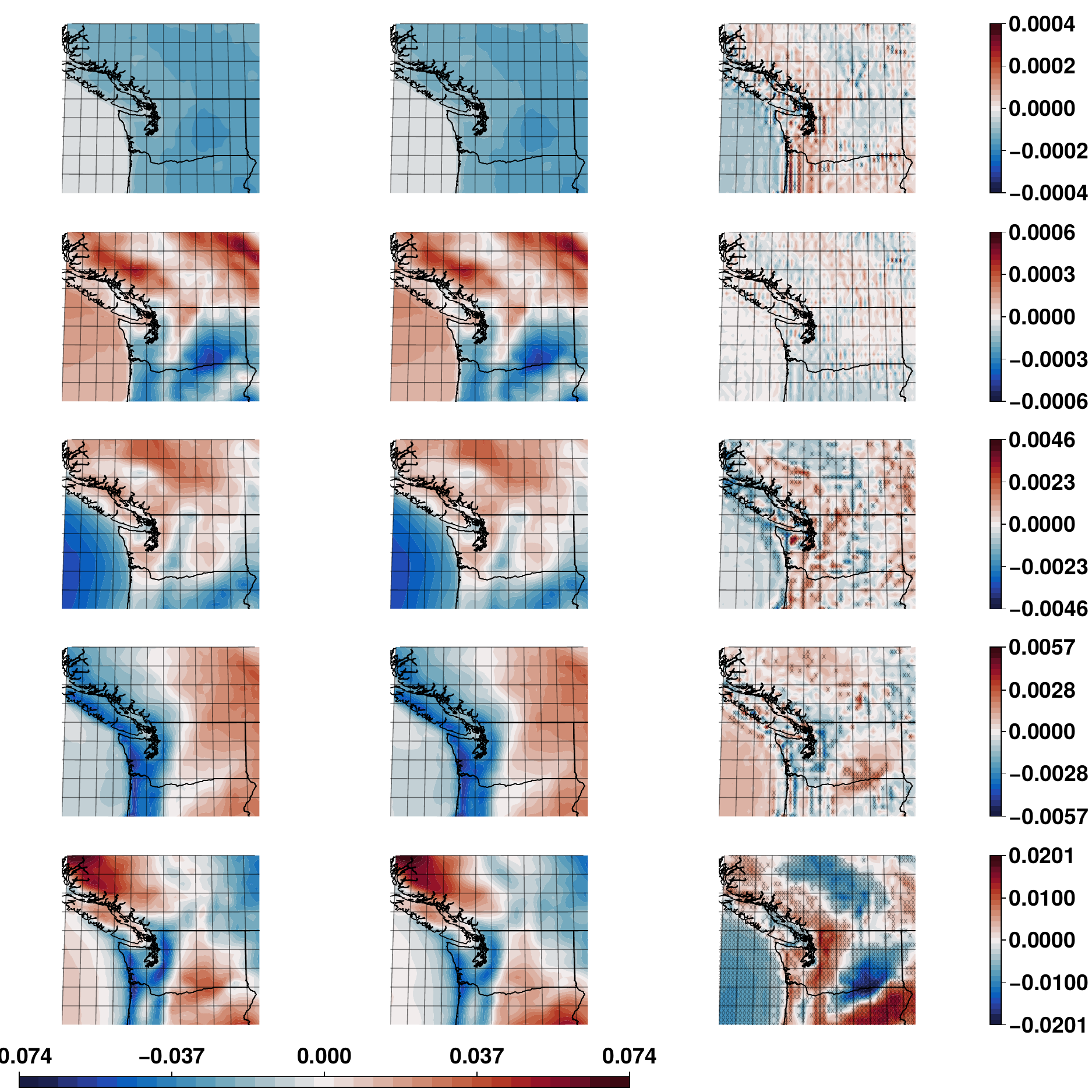}
    \caption{Estimated spatial basis functions $\vu_1$ (top) through $\vu_{5}$ (bottom). The left column are the estimates from the deterministic SVD, the middle column are the posterior means, and the right column are the posterior difference between the posterior mean and the algorithmic estimate where locations whose 95\% credible interval does not cover zero are denoted with an `x'.} 
    \label{fig:U15basis}
\end{figure}

\begin{figure}
    \centering
    \includegraphics[width = 0.9\linewidth]{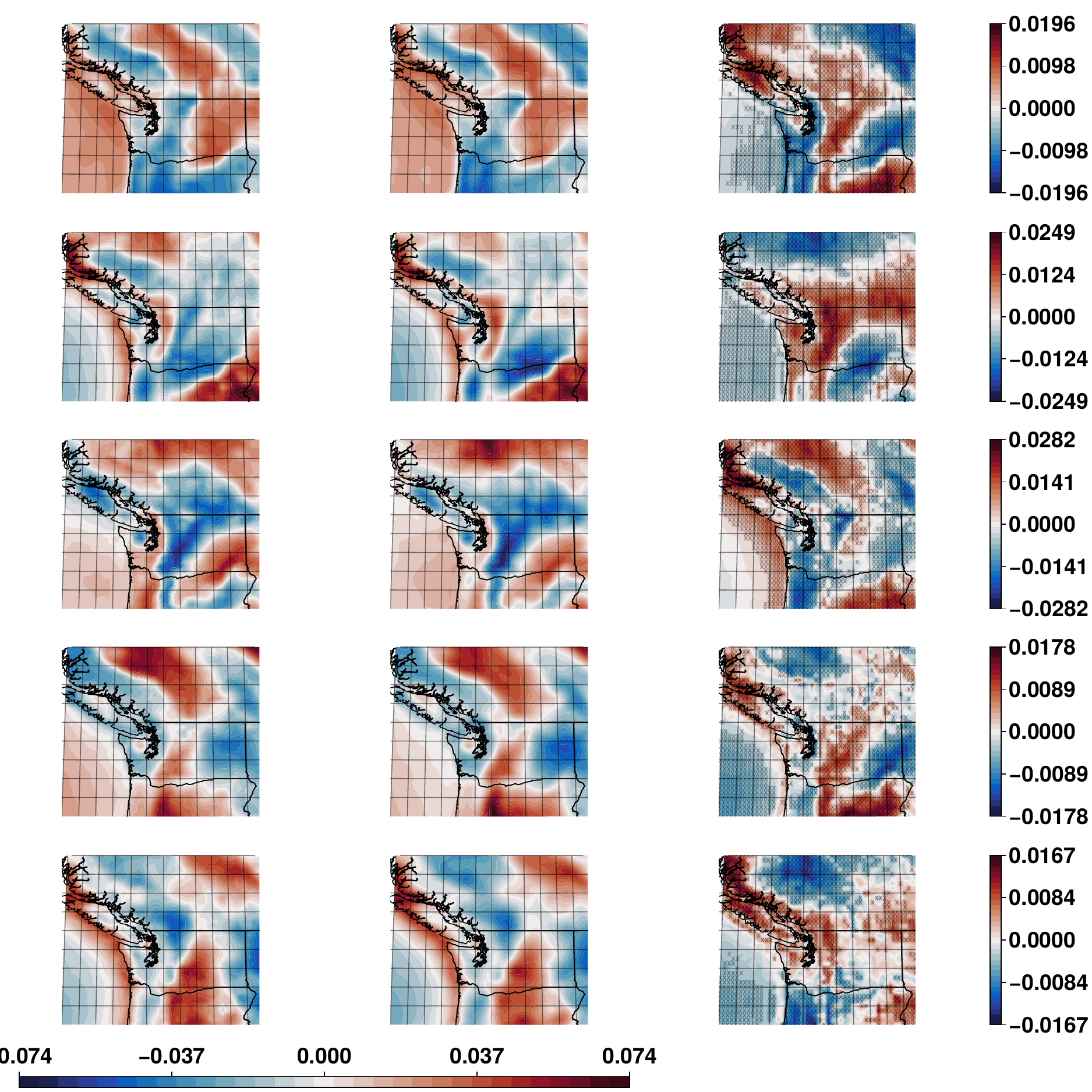}
    \caption{Estimated spatial basis functions $\vu_6$ (top) through $\vu_{10}$ (bottom). The left column are the estimates from the deterministic SVD, the middle column are the posterior means, and the right column are the posterior difference between the posterior mean and the algorithmic estimate where locations whose 95\% credible interval does not cover zero are denoted with an `x'.} 
    \label{fig:U610basis}
\end{figure}

\begin{figure}
    \centering
    \includegraphics[width = 0.9\linewidth]{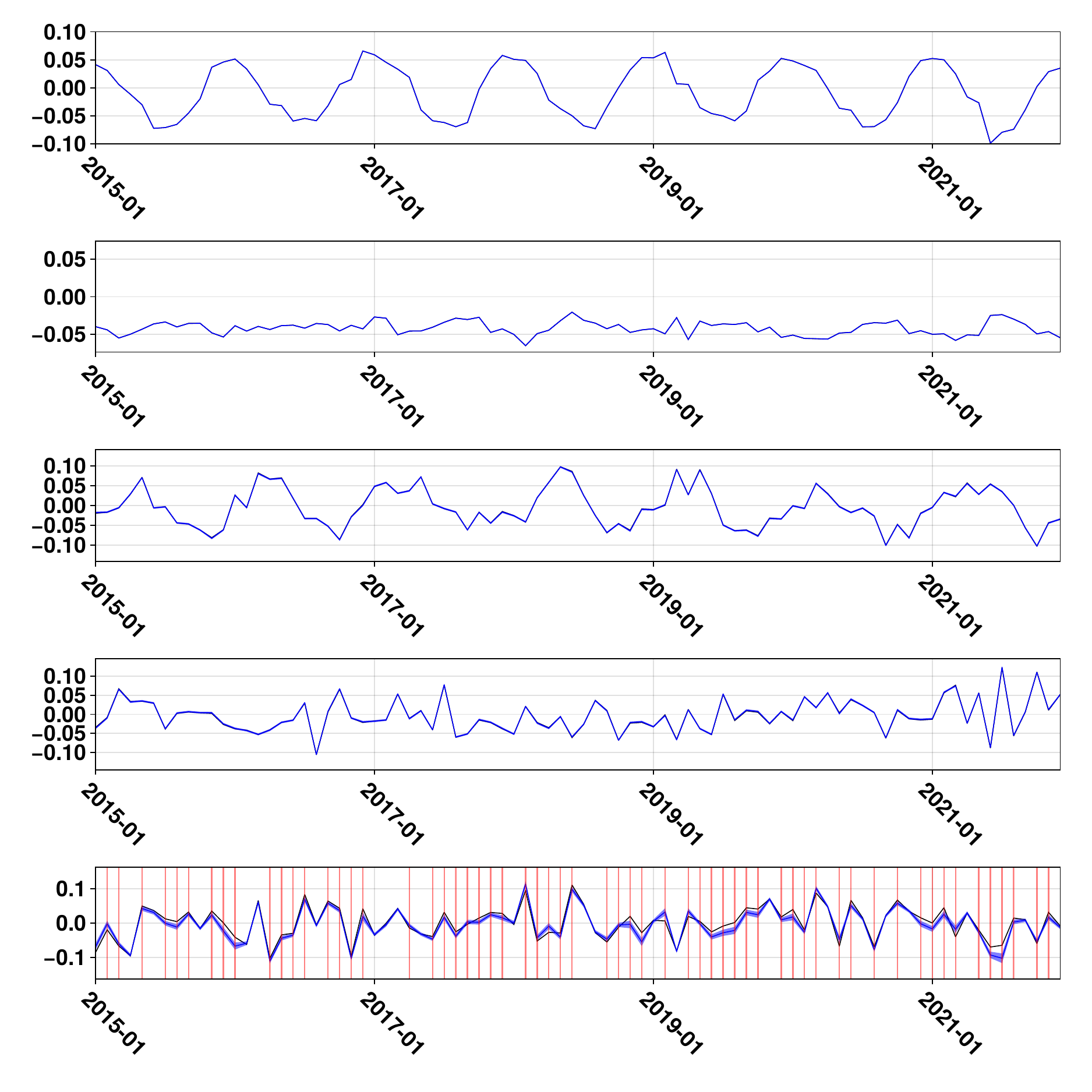}
    \caption{Estimated temporal basis functions $\vv_1$ (top) through $\vv_5$ (bottom) from January 2010 to December 2021. The black line is the algorithmic estimate, the solid blue line is the posterior mean, and the blue shaded regions are the 95\% credible intervals. Because it is difficult to see the differences, time points where the 95\% credible interval does not cover the deterministic estimate are marked with a vertical line.}
    \label{fig:V15basis}
\end{figure}

\begin{figure}
    \centering
    \includegraphics[width = 0.9\linewidth]{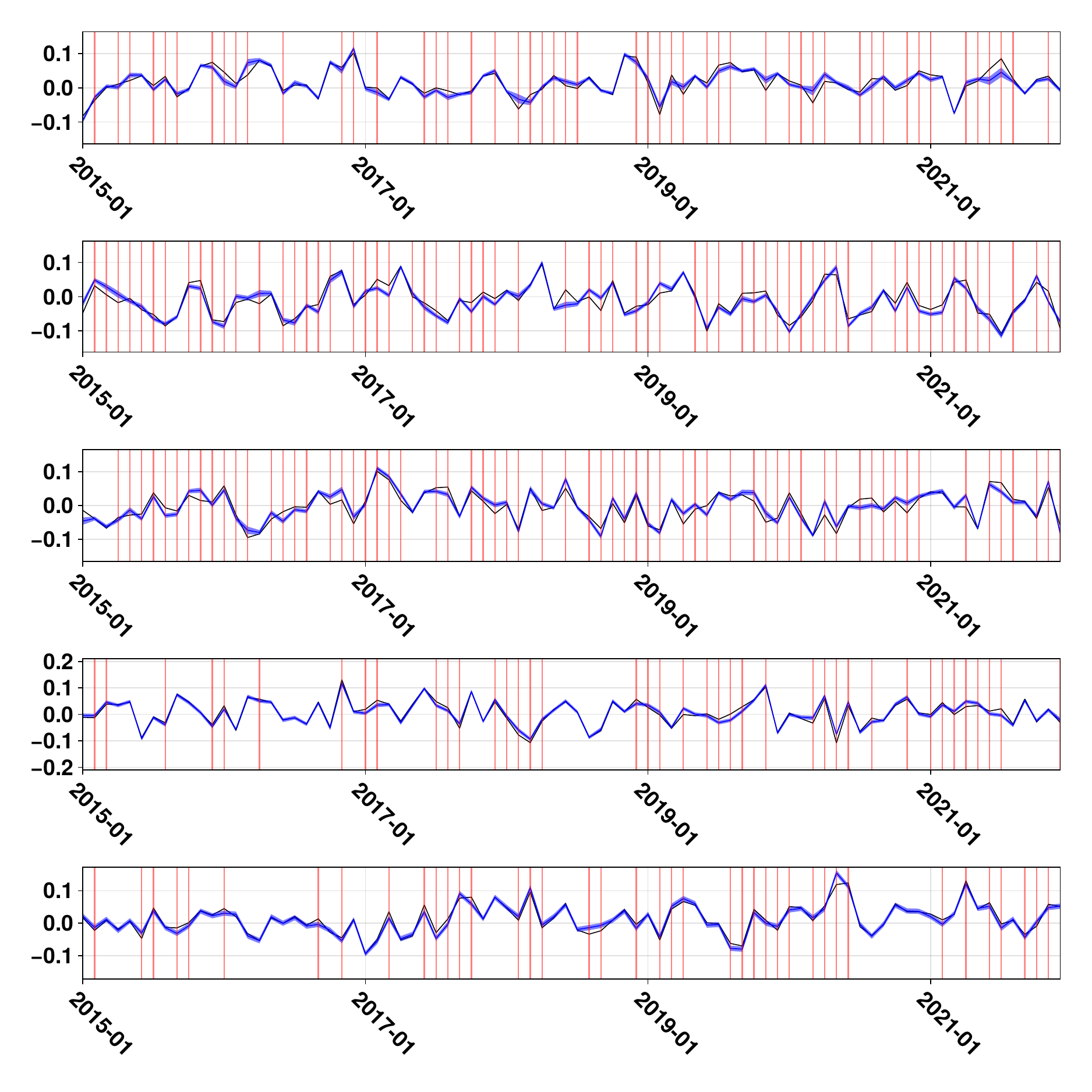}
    \caption{Estimated temporal basis functions $\vv_6$ (top) through $\vv_{10}$ (bottom) from January 2010 to December 2021. The black line is the algorithmic estimate, the solid blue line is the posterior mean, and the blue shaded regions are the 95\% credible intervals. Because it is difficult to see the differences, time points where the 95\% credible interval does not cover the deterministic estimate are marked with a vertical line.}
    \label{fig:V610basis}
\end{figure}

\end{document}